\documentclass[journal]{IEEEtran}

\usepackage{xcolor,soul,framed} 
\colorlet{shadecolor}{yellow}
\usepackage[pdftex]{graphicx}
\graphicspath{{../pdf/}{../jpeg/}}
\DeclareGraphicsExtensions{.pdf,.jpeg,.png}

\usepackage{amsthm}
\usepackage{amssymb}
\usepackage{array}
\usepackage{mdwmath}
\usepackage{mdwtab}
\usepackage{eqparbox}
\usepackage{url}

\usepackage{cite}
\usepackage{amsmath,amssymb,amsfonts}
\usepackage{algorithm}
\usepackage{algorithmic}
\usepackage{graphicx}
\usepackage{textcomp}
\usepackage{svg}
\usepackage{mathtools}

\usepackage{mathrsfs}
\usepackage{siunitx}
\usepackage{booktabs}
\usepackage{multirow, multicol}

\newlength{\cmidruleWidth} 
\newtheorem{definition}{Definition}
\newtheorem{problem}{Problem}
\newtheorem{theorem}{Theorem}
\newtheorem{game}{Game}
\newtheorem{remark}{Remark}

\newcommand{\set}[1]{\left\{#1\right\}}

\DeclareMathOperator{\argmin}{argmin}
\DeclareMathOperator*{\argminlimits}{argmin}

\begin{document}

\bstctlcite{IEEEexample:BSTcontrol}

\title{
Hypergame-based Cognition Modeling and Intention Interpretation for Human-Driven Vehicles in Connected Mixed Traffic 
}

\author{
Jianguo Chen, Zhengqin Liu,  Jinlong Lei, Peng Yi, Yiguang Hong, \IEEEmembership{Fellow,~IEEE}, Hong Chen, \IEEEmembership{Fellow,~IEEE}%
\thanks{ Jianguo Chen is with the State Key Laboratory of Mathematical Sciences, Academy of Mathematics and Systems Science, Chinese Academy of Sciences, Beijing 100190, China, and also with the University of Chinese Academy of Sciences, Beijing 100049, China
{\tt\small (chenjianguo@amss.ac.cn)}}%
\thanks{ Zhengqin Liu and Hong Chen are with the Department of Control Science and Engineering, Tongji University, Shanghai 201804, China {\tt\small (2230709@tongji.edu.cn, chenhong2019@tongji.edu.cn)}}%
\thanks{
Jinlong Lei, Peng Yi  and Yiguang Hong are with the Department of Control Science and Engineering, Tongji University, Shanghai, 201804, China; also with State Key Laboratory of Autonomous Intelligent Unmanned Systems, and Frontiers Science Center for Intelligent Autonomous Systems, Ministry of Education, and the Shanghai Institute of Intelligent Science and Technology, Tongji University, Shanghai, 200092, China {\tt\small (leijinlong@tongji.edu.cn, yipeng@tongji.edu.cn, yghong@tongji.edu.cn)}
}%
\thanks{
Jianguo Chen and Zhengqin Liu contributed equally to this work.
}
\thanks{
This work has been submitted to the IEEE for possible publication. Copyright may be transferred without notice, after which this version may no longer be accessible.
}
\thanks{
© 2025 IEEE. Personal use of this material is permitted. Permission from IEEE must be obtained for all other uses, including reprinting/republishing this material for advertising or promotional purposes, collecting new collected works for resale or redistribution to servers or lists, or reuse of any copyrighted component of this work in other works.
}
}


\maketitle

\begin{abstract}
With the practical implementation of connected and autonomous vehicles (CAVs), the traffic system is expected to remain a mix of CAVs and human-driven vehicles (HVs) for the foreseeable future. To enhance safety and traffic efficiency, the trajectory planning strategies of CAVs must account for the influence of HVs, necessitating accurate HV trajectory prediction. Current research often assumes that human drivers have perfect knowledge of all vehicles' objectives, an unrealistic premise. This paper bridges the gap by leveraging hypergame theory to account for cognitive and perception limitations in HVs. We model human bounded rationality without assuming
them to be merely passive followers and propose a hierarchical cognition modeling framework that captures cognitive relationships among vehicles. We further analyze the cognitive stability of the system, proving that the strategy profile where all vehicles adopt cognitively equilibrium strategies constitutes a hyper Nash equilibrium when CAVs accurately learn HV parameters (Theorem \ref{thm:hne}). To achieve this, we develop an inverse learning algorithm for distributed intention interpretation via vehicle-to-everything (V2X) communication, which extends the framework to both offline and online scenarios. Additionally, we introduce a distributed trajectory prediction and planning approach for CAVs, leveraging the learned parameters in real time. Simulations in highway lane-changing scenarios demonstrate the proposed method's accuracy in parameter learning, robustness to noisy trajectory observations, and safety in HV trajectory prediction. The results validate the effectiveness of our method in both offline and online implementations.
\end{abstract}


%
\IEEEpeerreviewmaketitle


\section{Introduction}
\label{sec:introduction}

\IEEEPARstart{W}{ith} the practical implementation of CAVs, the traffic system is expected to remain a mix of CAVs and HVs for the foreseeable future \cite{LI2023104258,10648806}. To ensure road safety and improve traffic efficiency, CAVs must have the ability to accurately predict the trajectories of HVs. This capability urgently requires interpreting human drivers' intentions. 

Rule-based and learning-based methods are commonly used in HV modeling approaches in previous studies. Rule-based methods, such as \cite{gipps1981behavioural, treiber2000congested, newell2002simplified}, model the driving strategies of HVs in traffic flow as maintaining a constant speed and following the lead vehicle according to given rules. These methods provide a computationally simple and analyzable modeling approach for HVs' behavior, making them the most commonly used method in mixed traffic studies. However, since the rules are overly simplified compared to the decision-making processes of human drivers in reality, these methods struggle to accurately simulate trajectories in complex situations. Unlike the analytically focused rule-based methods, learning-based methods such as deep learning \cite{gao2023dual}, reinforcement learning \cite{zhang2018human,10541090}, and imitation learning \cite{9990591} learn driving strategies of human drivers directly from datasets of real HVs' trajectories. Due to the typically higher model complexity and a greater number of parameters in learning-based methods than in rule-based methods, they possess the capability to generate more complex driving behaviors. These methods are also frequently used to enable CAVs to make human-like decisions. However, both rule-based and learning-based methods lack consideration for the interaction patterns between HVs and CAVs. 

Thus, in this paper, we focus on game-theoretic methods \cite{yu2020multi}, modeling the decision-making processes of HVs and CAVs as a game problem. The decisions, i.e., the equilibrium of the game, are influenced by the utility functions and constraints of all vehicles, thereby explicitly constructing the impact of interactions. Recently, game-theoretic modeling of vehicle decision-making and interaction has gained increasing research attention, with advancements in the intention interpretation of agents in games. For example, \cite{mehr2023maximum} proposed the entropic cost equilibrium to characterize bounded rational decision-making in human interaction, and developed a maximum entropy inverse dynamic game algorithm to learn players' objective functions from trajectory datasets. In addition, \cite{peters2023online} proposed an intention interpretation algorithm based on a least-squares problem with Nash equilibrium constraints to calculate players' goals, state estimations, and trajectory predictions online. Most existing game-theoretic methods share a common flaw: they assume that human drivers understand the true objective functions of all HVs and CAVs. Yet, in reality, HVs do not precisely recognize CAVs' intentions \cite{9965180}. In previous studies considering the bounded rationality of HVs within game-theoretic frameworks, HVs are typically assumed to act as followers, reacting to the strategies of autonomous vehicles (AVs). For instance, in \cite{DI2020102710}, the AVs were modeled as the leader, while HVs were treated as followers. Similarly, in \cite{10127583}, brain-inspired modeling was employed to characterize HV behavior; however, the inputs to this model, such as trajectory tracking error and other observable external information, were predefined based on observed data.





Therefore, we extend this framework to a setting that accounts for the bounded rationality of HVs without assuming them to be merely passive followers. Faced with HVs with bounded rationality, CAVs need to identify the intentions of HVs through interactive trajectories so that they can plan trajectories more safely and efficiently. Because of the limited rationality of HVs and the uncertainty of CAVs about HVs' intentions, HVs and CAVs engage in a game based on their respective cognition rather than the same one, leading to a hypergame problem. Hypergame theory extends the traditional game theory to account for conflicts involving misperceptions. It allows for a game model incorporating differing perspectives, representing variations in each player's information, beliefs, and understanding of the game \cite{kovach2015hypergame,cheng2022misperception}. Based on the hypergame framework, this paper clearly characterizes the multi-level cognitive structure between HVs and CAVs. Then a Karush-Kuhn-Tucker (KKT)-based inverse game algorithm is proposed to estimate parameters in the objective functions of HVs. Subsequently, we design a collaborative intention interpretation mechanism between CAVs and the roadside unit (RSU), which coordinates computation via V2X communication. Finally, we conduct multiple simulations in highway lane-changing scenarios to evaluate the accuracy and safety of the proposed method. The main contributions of this paper are as follows:
\begin{itemize}
\item We model human bounded rationality by incorporating cognitive and perception limitations, and design a hierarchical cognition modeling framework using hypergames. This framework can effectively characterize the cognitive relationships among vehicles and their impact on decision-making processes.
\item We analyze the cognitive stability of vehicles by proving that the strategy profile, where all vehicles adopt cognitively optimal responses, constitutes a hyper Nash equilibrium when CAVs successfully learn the true parameters of HV (Theorem \ref{thm:hne}).
\item We propose inverse game-theoretic methods for distributed and vehicle-road collaborative intention interpretation, addressing both offline and online scenarios. Leveraging the hierarchical cognition model, we further develop a distributed trajectory prediction and planning process for CAVs.
\item Using simulations in both offline and online scenarios, we demonstrate the proposed method's robustness in parameter learning and its effectiveness in ensuring accurate and safe trajectory prediction, even under noisy observation conditions.
\end{itemize}

\textit{Notation}: $\boldsymbol{0}$ represents a zero vector; The operator $\operatorname{vec}(a_1,\dots,a_l)$ means joining column vectors or scalars $a_1,\dots,a_l$ into a vector $(a_1^\top,\dots,a_l^\top)^\top$; For a vector $x$ and a matrix $A$, $\|x\|^2_A=x^\top Ax$; $a\circ b$ denotes the Hadamard product of vectors $a$ and $b$; $[n]$ denotes the set $\set{1,\dots,n}$; The symbol $\oplus$ denotes the direct sum operation, which combines two matrices into a block diagonal matrix. To help readers, the frequently used symbols in this
article are listed in Table~\ref{tab:notation}.

\begin{table}[ht]
\label{tab:notation}
\caption{Explanation of Symbols}
\centering
\begin{tabular}{c | >{\centering\arraybackslash}m{6cm}} 
\hline
\textbf{Notation} & \textbf{Meaning} \\ \hline
$\mathcal{C}, \mathcal{N}$ & The sets of all connected and autonomous vehicles, and all vehicles, respectively. \\ \hline
$s_{i}$ & Decision variables for vehicle $i$. \\ \hline
$x_{\text{ref},i}$ & The reference trajectory of vehicle $i$. \\ \hline
$s_{i,j}$ & Decision variables of vehicle $i$ in vehicle $j$'s cognition. \\ \hline
$s_{(i,j),l}$ & Decision variables of vehicle $i$ as perceived by vehicle $j$, where vehicle $j$'s perception is further understood by vehicle $l$. \\ \hline
$s_{0,C}$ & HV’s strategy as perceived by CAVs. \\ \hline
\(\mathbf{s}_{-i}, \mathbf{s}_{\neg i}\) & The strategy profile of all other vehicles except \(i\); the strategy profile of all other CAVs for a CAV \(i \in \mathcal{C}\).\\ \hline
$\theta_i,\theta$ & The parameter vector of vehicle $i$, encoding weights from $Q_i$ and $R_i$;  the parameter vector for all vehicles, $\operatorname{vec}(\theta_i, i \in \mathcal{N})$.\\ \hline
$\theta_{j,i}$ & Vehicle $i$'s estimation of parameter $\theta_j$. \\   \hline
$\theta_{0,C}$ & HV’s parameter as perceived by CAVs. \\ \hline
$S_i(\mathbf{s}_{-i})$ & The strategy set of vehicle $i \in \mathcal{N}$, depending on other vehicles. \\ \hline
$J_i(s_i; \theta_i)$ & The objective function of vehicle $i$, representing its optimization target. \\ \hline
$h_i, g_i$ & Equality and inequality constraints for vehicle \(i\), respectively. \\ \hline
$\theta_{i,\text{true}}, \theta_{i,\text{ave}}$ & The true weight parameter and its average value for vehicle $i$, respectively. \\  \hline
$\epsilon_c, \epsilon_p$ & The cognitive threshold; the perceptual threshold. \\   \hline  
$G_{\text{true}}$ & The actual game shared by all players. \\   \hline
$G_{\text{true},i}$ & Vehicle $i$'s perception of the actual game $G_{\text{true}}$. \\   \hline
$\prescript{0}{}{H}$ & The level 0 hypergame, representing the game without misperceptions. \\   \hline
$\prescript{1}{}{H}$ & The level 1 hypergame, capturing subjective views of all players. \\   \hline  
$^1H_i$ & Level 1 hypergame perceived by vehicle $i$. \\ \hline
$^2H$ & Level 2 hypergame incorporating all players’ perceptions. \\ \hline
$\mathcal{T}_t$ & The time segment for time period $t$, where $\mathcal{T}_t = \{k_{t-1}, \dots, k_t\}$. \\ \hline
$G^t$ & The dynamic game during time period $t$. \\ \hline
$s_i^t$ & Strategy of vehicle $i$ during time period $t$. \\ \hline
$\theta_{0,C}^t$ & The CAVs' estimate of HV's parameter $\theta_{0,\text{true}}$ at time $t$. \\ \hline
$\hat{s}_0^{t-1}$ & Observed trajectory of HV in time period $t-1$. \\ \hline
$\tau$ & Number of sequential time periods in the prediction horizon. \\  \hline
$s_{0,C}^t$ & Predicted trajectory of HV by CAVs in period $t$. \\ \hline
\end{tabular}
\end{table}

\section{Trajectory Planning Games}
\label{sec:model_of_game}

We consider a road traffic scenario involving an RSU in the absence of traffic signals, where CAVs dominate the traffic system, while HVs are scarce. In this setup, all CAVs and the RSU communicate seamlessly through V2X technology, whereas HVs lack this communication capability \cite{OlaverriMonreal2016HumanFI}. In this section, we model the trajectory interactions between vehicles using game theory, formulating the problem as a Generalized Nash Equilibrium Problem. The objective function and strategy constraints of the model are explicitly defined. The proposed approach aligns with the framework presented in \cite{liu2024semi}, where similar game-theoretic methods are employed to model multi-agent interactions.   

We focus on the interaction patterns between HVs and nearby CAVs. Given the local dominance of CAVs, we specifically consider the most common scenario that involves a single HV interacting with multiple CAVs. Accordingly, this paper primarily investigates the interaction between one HV and \(n\) CAVs. Let \(\mathcal{C} = \{1, 2, \dots, n\}\) represent the set of CAVs, and \(\{0\}\) represent the HV. Then the set of all vehicles, \(\mathcal{C} \cup \{0\} = \{0, 1, \dots, n\}\), is denoted by \(\mathcal{N}\). Figure~\ref{fig:scenario_example} illustrates an example of this scenario, where the trajectories of the HV and CAVs are depicted as curves, and their predicted positions at five discrete future time steps are marked by dots.

\begin{figure}[H]
    \centering
    \includegraphics[width=\linewidth, keepaspectratio]{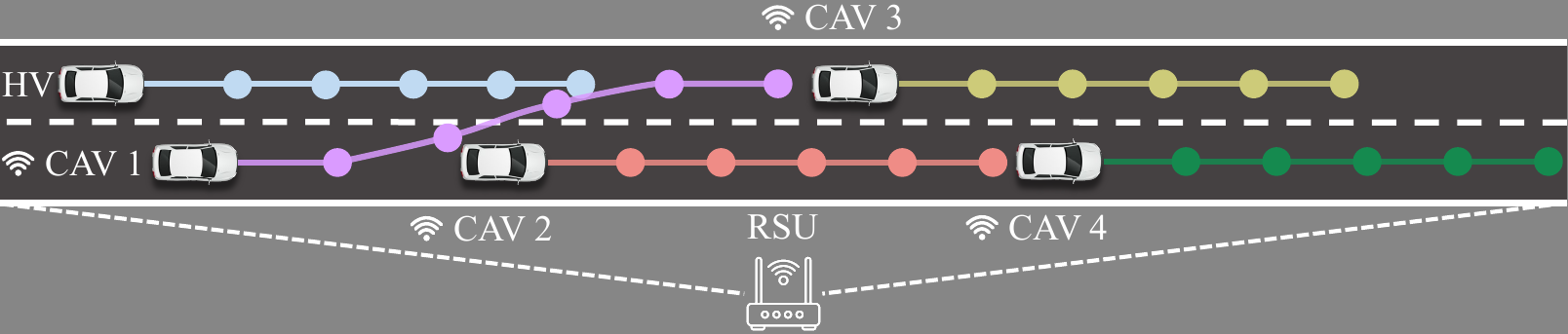}
    \caption{An example of the traffic scenario involving the interaction of one HV and four CAVs on a three-lane road. The trajectories of the vehicles are color-coded to represent their respective paths, while an RSU supports the coordinated maneuvers of CAVs.}
    \label{fig:scenario_example}
\end{figure}

\subsection{Objective Function}\label{subsec:obj_fun}

In this paper, we employ the widely used bicycle model as the basis for vehicle dynamics modeling \cite{Chen2021ANL,matute2019experimental}. The analysis is conducted in a discrete-time framework. Let \(\mathcal{T}\) denote the set of discrete time steps. For each vehicle \(i \in \mathcal{N}\), the state-control pair at time step \(k \in \mathcal{T}\) is denoted as \(s_i(k) = \operatorname{vec}(x_i(k), u_i(k))\), where \(x_i(k)\) represents the state variables and \(u_i(k)\) represents the control variables. The state vector is defined as \(x_i(k) = [p_{x,i}(k), p_{y,i}(k), v_i(k), \psi_i(k)]^\top\), encompassing the vehicle’s position, velocity, and heading angle. The control vector is given by \(u_i(k) = [a_i(k), \delta_i(k)]^\top\), which includes the acceleration and front wheel steering angle.  

Over the time horizon \(\mathcal{T}\), the complete strategy of vehicle \(i\) is represented as \(s_i = \operatorname{vec}(s_i(k), k \in \mathcal{T})\), excluding the initial state and terminal control at the boundaries of \(\mathcal{T}\). The strategy profiles of all other vehicles except \(i\) are denoted as \(\mathbf{s}_{-i} = \operatorname{vec}(s_j, j \in \mathcal{N} \backslash \{i\})\). For a CAV \(i \in \mathcal{C}\), the strategy profiles of other CAVs are denoted as \(\mathbf{s}_{\neg i} = \operatorname{vec}(s_j, j \in \mathcal{C} \backslash \{i\})\). The strategy set of vehicle \(i \in \mathcal{N}\) is denoted as \(S_i(\mathbf{s}_{-i})\), which depends on the strategies of other vehicles. Each vehicle aims to minimize its objective function \(J_i(s_i; \theta_i)\), subject to the feasible strategy set \(S_i(\mathbf{s}_{-i})\):  

\begin{equation}
    \begin{aligned}
        J_i(s_i; \theta_i) = \frac{1}{2}\sum_{k \in \mathcal{T}} \Big( & \|x_i(k+1) - x_{\text{ref},i}(k+1)\|^2_{Q_i} \\
        &+ \|u_i(k)\|^2_{R_i} \Big),
    \end{aligned}
    \label{eq:objective}
\end{equation}  
where \(x_{\text{ref},i}\) represents the reference trajectory of vehicle \(i\), and \(Q_i\) and \(R_i\) are diagonal positive definite weighting matrices for the state deviation and control effort, respectively.  

The parameter vector \(\theta_i = \operatorname{vec}\left(\operatorname{diag}(Q_i), \operatorname{diag}(R_i)\right)\) encodes the weights associated with \(Q_i\) and \(R_i\), characterizing the driving style of vehicle \(i \in \mathcal{N}\). The set of all possible parameter values is denoted by \(\Theta\), which is assumed to be bounded to ensure the driving style parameters remain within a finite and realistic range. Specifically, each \(\theta_i \in \Theta\) satisfies \(\theta_{\text{min}} \leq \theta_i \leq \theta_{\text{max}}\), where \(\theta_{\text{min}} > 0\) is the lower bound and \(\theta_{\text{max}}\) is the upper bound. For the entire system, the driving style parameters for all vehicles are collectively represented as \(\theta = \operatorname{vec}(\theta_i, i \in \mathcal{N})\). 

In this study, each CAV \(i \in \mathcal{C}\) is capable of directly sharing its decision variable \(s_i\) and reference trajectory \(x_{\text{ref},i}\) with other CAVs and the RSU. However, to safeguard the proprietary aspects of its trajectory planning algorithm, the weight parameter \(\theta_i\), which determines \(i\)'s driving behavior and style, is kept private and not shared. The estimation of the HV's reference trajectory \(x_{\text{ref},0}\) is beyond the scope of this work. Instead, we assume that the final target state of the HV is known to the CAVs. This assumption was widely used in related studies \cite{mehr2023maximum, peters2023online, fang2024cooperative}. The reference trajectory for the HV is generated using the same method applied to CAVs. Consequently, the objective function \(J_i\) for each CAV \(i\) is fully determined by its weight parameter \(\theta_i\). The true weight parameter of each vehicle \(i\) is denoted as \(\theta_{i,\text{true}} \in \Theta\).

\subsection{Constraints}\label{subsec:cons}

Next, we define the constraints \(S_i ,i\in \mathcal{N}\). These constraints incorporate both vehicle dynamics and safety requirements. The constraints include the following categories:

{\it (1) Dynamics Constraints}: The dynamics constraints are modeled using the bicycle model, as described in \cite{matute2019experimental}. The states of each vehicle include its position, velocity, and heading angle, while the controls consist of acceleration and front-wheel steering angle. Let \(L\) represents the vehicle length. The continuous-time dynamics are expressed as:
\begin{equation}
\left\{
  \begin{aligned}
      &\dot{p}_{x,i} = v_i\cos\psi_i, \\
      &\dot{p}_{y,i} = v_i\sin\psi_i, \\
      &\dot{v}_i = a_i, \\
      &\dot{\psi}_i = \frac{v_i\tan\delta_i}{L}.
  \end{aligned}
\right.
\label{eq:continuous_time_dynamics_equation}
\end{equation}
To ensure computational tractability, we adopt the linearized discrete-time approximation of \eqref{eq:continuous_time_dynamics_equation} as the dynamics constraints. This approximation maintains the model's fidelity while enabling efficient optimization.

{\it (2) Box Constraints}: The physical capabilities of each vehicle impose limits on its states and controls. Specifically, the velocity, acceleration, and front-wheel steering angle of vehicle \(i\) are constrained as follows:
\begin{equation*}
    \begin{aligned}
        & v_{i,\min} \leq v_i(k) \leq v_{i,\max}, \\
        & a_{i,\min} \leq a_i(k) \leq a_{i,\max}, \\
        & \delta_{i,\min} \leq \delta_i(k) \leq \delta_{i,\max}.
    \end{aligned}
\end{equation*}
These bounds ensure the feasibility and safety of vehicle behaviors under real-world operating conditions.

{\it (3) Lane constraints}: We set the constraint that the four vertices of a slightly larger concentric rectangle of the vehicle's plain view must be within the lane to prevent the vehicle from crossing the lane lines. Denote the rectangle's length and width as $L_E$ and $W_E$. The two-dimensional homogeneous coordinates of the rectangle vertex at the front left of the vehicle (denoted as point $A$) at time $k$ are
\begin{equation*}
    \begin{aligned}
        \tilde{\boldsymbol{p}}_{A,i}(k)=\big(& p_{x, i}(k)+\frac{L_E}{2} \cos (\psi_{i}(k))-\frac{W_E}{2} \sin (\psi_{i}(k)), \\
         & p_{y, i}(k)+\frac{L_E}{2} \sin (\psi_{i}(k))+\frac{W_E}{2} \cos (\psi_{i}(k)), 1 \big)^\top.
    \end{aligned}
\end{equation*}
Let $\ell\in\mathscr{L}$ denote the lane boundary index. At each time $k$, the lane boundary $\Gamma_\ell$ is linearized, i.e., a tangent is taken at the projection point of vehicle $i$'s position. Let the tangent's coefficients be $\boldsymbol{a}_{\ell,i}(k)=(a_{\ell,i}(k),b_{\ell,i}(k),c_{\ell,i}(k))$. Considering the positions of the four vertices $A, B, C, D$ of the rectangle, the lane constraint for vehicle $i$ at time $k$ is represented as
\begin{equation*}
    \begin{aligned}
        m_{\ell,i}(s_{i}(k))=(& \allowbreak \boldsymbol{a}_{\ell,i}(k) \tilde{\boldsymbol{p}}_{A,i}(k), \allowbreak \boldsymbol{a}_{\ell,i}(k) \tilde{\boldsymbol{p}}_{B,i}(k), \\ 
        & \boldsymbol{a}_{\ell,i}(k) \tilde{\boldsymbol{p}}_{C,i}(k), \allowbreak \boldsymbol{a}_{\ell,i}(k) \tilde{\boldsymbol{p}}_{D,i}(k) \allowbreak )^\top \allowbreak \leq \boldsymbol{0}.
    \end{aligned}
\end{equation*}

{\it (4) Collision avoidance constraints}: Let the vehicle width be $W$ and the diagonal length of the vehicle's plain view rectangle be $D$. The collision avoidance range is set as a super-ellipse $\frac{x^6}{(L / 2 + D / 2) ^ 6} + \frac{y ^ 6}{(W / 2 + D/ 2) ^ 6} = 1$. The coordinates of vehicle $j$ at time step $k$ in the reference frame with the center of vehicle $i$ as the origin and the direction of the vehicle's head as the $x$-axis are denoted as $(\check{p}_{x,j}(k),\check{p}_{y,j}(k))$. At this moment, the collision avoidance constraint of vehicle $i$ on vehicle $j$ is
\[
h_{i,j}(s_i(k),s_j(k))=1-\frac{(\check{p}_{x,j}(k))^6}{(\frac{L}{2} + \frac{D}{2}) ^ 6} - \frac{(\check{p}_{y,j}(k)) ^ 6}{(\frac{W}{2} + \frac{D}{2}) ^ 6}\leq 0.
\]

{\it (5) Driving behavior constraints}: We only impose driving behavior constraints on straight-driving and lane-changing vehicles. For a straight-driving vehicle $i$, the unit vector along the center line of its lane in the direction of vehicle $i$'s movement is denoted as $\boldsymbol{d}=(d_x,d_y)$. We impose an equality constraint that the heading angle must align with the direction of $\boldsymbol{d}$: $\psi_i(k)-\operatorname{atan2}(d_y,d_x)=0, \forall k = \mathcal{T}$. For a lane-changing vehicle $i$, its homogeneous coordinates at time $k$ are $\tilde{\boldsymbol{p}}_{i}(k)=\operatorname{vec}(p_{x,i}(k),p_{y,i}(k),1)$. The coefficients of the center line of its lane are denoted as $\boldsymbol{a}_{\ell_c,i}(k)=(a_{\ell_c,i}(k),b_{\ell_c,i}(k),c_{\ell_c,i}(k))$, where $\ell_c \in \mathscr{L}_c$ is the index of the lane center line. We constrain that during the lane-changing process, vehicle $i$ must be on the side of its lane center line closer to the target lane: $\boldsymbol{a}_{\ell_c,i}(k) \tilde{\boldsymbol{p}}_{i}(k) \leq 0$. This constraint ensures that vehicle $i$ avoids unnecessary opposite-direction steering during the lane-changing process.
\begin{remark}\label{re:line}
    For simplicity, all nonlinear constraints are linearized by retaining only the first-order terms in their Taylor expansion. The detailed linearization procedures are same as those outlined in \cite{liu2024semi}.
\end{remark}
Under Remark~\ref{re:line}, the set of constraints for vehicle \(i\) at time period \(\mathcal{T}\) can be compactly expressed as:
\begin{equation}\label{eq:cons}
    S_{i}\left(\mathbf{s}_{-i}\right)=\left\{s_{i} \mid h_{i}(s_{i},\mathbf{s}_{-i}) = \mathbf{0}, \, g_{i}(s_{i},\mathbf{s}_{-i}) \leq \mathbf{0} \right\},
\end{equation}
where \(h_i\) represents the linear equality constraints, and \(g_i\) denotes the linear inequality constraints. These constraints ensure the feasibility of the vehicle’s trajectory within the given operational limits.

\subsection{Game Model} \label{sec:game_model}
We model the interaction among vehicles as a generalized Nash equilibrium problem (GNEP), where each vehicle's strategy set depends on the strategies of the other vehicles \cite{facchinei2010generalized}. This interdependence arises from the coupled constraints, which reflect the joint influence of all vehicles in the system.

The game without misperceptions is formally defined as follows:  

\begin{game}  
The trajectory planning game without misperceptions between the HV and CAVs is represented by:  
\[
G_{\text{true}} = \left( \mathcal{N}, \{S_{i}(\mathbf{s}_{-i})\}_{i \in \mathcal{N}}, \{J_{i}(s_i; \theta_{i,\text{true}})\}_{i \in \mathcal{N}} \right),  
\]
where \(S_{i}(\mathbf{s}_{-i})\) represents the strategy set of vehicle \(i\), which depends on the strategies of all other vehicles \(\mathbf{s}_{-i}\) as defined in \eqref{eq:cons}, and \(J_{i}(s_i; \theta_{i,\text{true}})\) is the objective function of vehicle \(i\) with respect to its true parameter \(\theta_{i,\text{true}}\), which depends on its own strategy \(s_{i}\), as defined in \eqref{eq:objective}.
\label{de:level_0_hypergame}  
\end{game}

Then we introduce the concept of a GNE in the following definition.

\begin{definition}  
A strategy profile \(\{s_{i}^*\}_{i \in \mathcal{N}}\) is a GNE of \(G_{\text{true}}\) in Game~\ref{de:level_0_hypergame} if, for each \(i \in \mathcal{N}\), the following condition holds:  
\begin{equation*}  
    J_{i}(s_{i}^*; \theta_{i,\text{true}}) \leq J_{i}(s_{i}; \theta_{i,\text{true}}), \quad \forall s_{i} \in S_{i}\left( \mathbf{s}_{-i}^* \right),  
\end{equation*}  
where \(\theta_{i,\text{true}}\) represents the true driving style parameter of vehicle \(i\).  
\end{definition}

In this formulation, the GNE captures the strategic interdependence of the vehicles by accounting for the coupled constraints in their strategy sets. At equilibrium, no vehicle can unilaterally adjust its strategy to achieve a lower value of its cost function \(J_{i}\), given the strategies of all other vehicles. This concept is particularly suitable for analyzing interactions in mixed traffic scenarios, where vehicles must consider both their own objectives and the actions of others.

\section{Modeling Cognitive Structures among Vehicles under Hypergames}
\label{subsec:cognitive_structure}

In this section, we introduce a human driver model that accounts for bounded rationality which reflects the cognitive and perceptual limitations inherent in human drivers, enabling a more realistic analysis of mixed traffic scenarios. Building upon the human driver model, we propose a cognitive hierarchy model based on hypergames to describe the interactions between CAVs and HV. This model introduces the concept of subjective rationalizable strategies for vehicle agents at different cognitive levels, as well as the notion of Hyper Nash Equilibrium, providing a theoretical framework for analyzing decision-making processes in mixed traffic environments.

\subsection{Human Model}
To characterize the bounded rationality of human drivers, we define two critical concepts: cognitive limitation and  perceptual limitation. These concepts are essential for constructing a hypergame framework, where human drivers operate based on subjective interpretations of the game rather than the true game structure. This discrepancy is the cornerstone of the multi-level hypergame model introduced in this study.

\subsubsection{Cognitive Limitation}  
Human drivers exhibit inherent cognitive constraints that limit their ability to fully comprehend and optimize the driving objective function. These constraints arise from the inability to precisely evaluate all relevant parameters, such as the exact weights in the objective function. Consequently, human drivers simplify complex strategies into generalized categories, such as aggressive or conservative driving styles, to better navigate the driving environment \cite{10440182}. This behavior is consistent with the concept of bounded rationality, wherein decision-making is based on approximate reasoning rather than precise optimization. Studies like Lindorfer et al. \cite{Lindorfer2018ModelingTI} demonstrate how human drivers face estimation errors in perceiving environmental variables such as spacing and relative velocity, reinforcing the notion of generalized approximations. Similarly, earlier research on bounded rationality in driving behavior \cite{Lubashevsky2003RationaldriverAI,Lubashevsky2002BoundedRD} further supports this perspective.  

In our model, HVs are assumed to recognize only the general driving styles of CAVs rather than the precise weights in their cost functions. Specifically, an HV's understanding of the true weight parameter \(\theta_{i,\text{true}}\) for vehicle \(i\) is represented by an approximate value, \(\theta_{i,\text{ave}}\), which corresponds to the average weight associated with the perceived driving style of vehicle \(i\). For instance, these driving styles—such as those illustrated in Figure~\ref{fig:styles}—may broadly categorize behaviors as aggressive or conservative. This approximation indicates that HVs generalize the true weights \(\theta_{i,\text{true}}\) into typical values \(\theta_{i,\text{ave}}\), reflecting their limited perception.  

We assume that these average weights, referred to as typical weights, are common knowledge shared among HVs and CAVs. To quantify this cognitive limitation, we define the cognitive threshold (\(\epsilon_c\)), which captures the maximum cognitive gap between the true driving style parameter and its approximation:  
\[
\epsilon_c = \max_{i \in \mathcal{C}} \| \theta_{i,\text{true}} - \theta_{i,\text{ave}} \|.
\]  
This metric reflects the degree of deviation introduced by human drivers' limited cognition and their reliance on approximations, as depicted in Figure~\ref{fig:styles}.

\begin{figure}
    \centerline{\includegraphics[width=\linewidth, keepaspectratio]{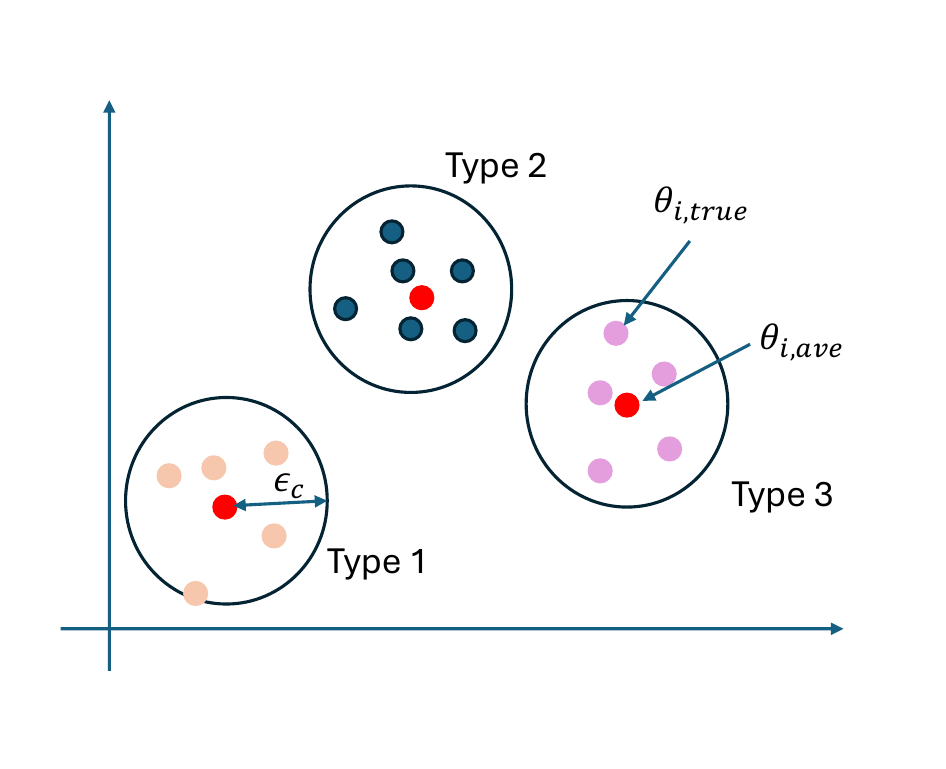}}
    \caption{Representation of Driving Styles, True and Average Weights, and the Cognitive Threshold (\(\epsilon_c\)). This illustrates the relationship between different driving styles (Type 1, Type 2, Type 3), their corresponding true weight parameters \(\theta_{i,\text{true}}\), the average weights \(\theta_{i,\text{ave}}\) associated with generalized driving styles, and the cognitive threshold (\(\epsilon_c\)). The red dots are \(\theta_{i,\text{ave}}\).}
    \label{fig:styles}
\end{figure}

\subsubsection{Perceptual Limitation}  \label{sec:Perceptual}  
Human drivers also exhibit perceptual limitations when responding to variations in their driving objective function. These limitations are characterized by insensitivity to small changes in stimuli, as supported by Lindorfer et al. \cite{Lindorfer2018ModelingTI}, who introduced the Enhanced Human Driver Model (EHDM). Their findings demonstrate that drivers tend to ignore minor perturbations in input stimuli unless these exceed a critical threshold, leading to threshold-driven decision-making. Wiedemann’s reaction sensitivity thresholds \cite{Wiedemann1974SIMULATIONDS} further support this behavior, describing how drivers respond only to perceptual changes that surpass specific thresholds.  

To model this limitation, we introduce the perceptual threshold (\(\epsilon_p > 0\)), which quantifies drivers' insensitivity to small variations in strategy efficacy. Formally, when the variation in the objective function value lies within the threshold \(\epsilon_p\),  
\[
J_{0}(s_{0}^* ; \theta_{0,\text{true}}) \leq J_{0}(s_{0} ; \theta_{0,\text{true}}) + \epsilon_p, \quad \forall s_{0} \in S_{0}\left( \mathbf{s}_{-0}^* \right),
\]  
an HV will not unilaterally deviate from its current strategy \(s_{0}^*\).  

This framework aligns with the concept of the \(\epsilon\)-Nash equilibrium, where deviations within \(\epsilon_p\) are considered negligible and do not impact decision-making. Studies like Noguchi et al. \cite{noguchi2007bayesian} and Miyazaki et al. \cite{miyazaki2013lambda} have demonstrated that agents with bounded rationality adapt and converge to \(\epsilon\)-Nash equilibria, which remain stable under slight perturbations. Similarly, Chen et al. \cite{chen2024approximation} proposed the notion of \(\epsilon\)-weakly Pareto-Nash equilibrium in multiobjective games, further capturing the effects of bounded rationality in decision-making.  

Empirical observations also support this modeling approach. For instance, Tan et al. \cite{Tan2022HumanMachineII} showed that drivers tend to disregard minor changes in stimuli, reacting only when changes exceed a noticeable threshold. Such findings reinforce the notion of a perceptual threshold, where small deviations are treated as inconsequential, ensuring stability in human drivers’ decision-making processes.  

\subsection{Hypergames}
For the HVs and CAVs sharing the same road, since they lack complete information about each other,  each of them has its own understanding of the game. Next, we present a framework for hierarchical hypergames based on the human model, along with the corresponding rationalizable strategies and the hyper Nash equilibrium. The cognitive structure of the HV and CAVs within the hypergame is illustrated in Fig. \ref{fig:cognition}. Now, we explain it in details.

\begin{figure}
    \centerline{\includegraphics[width=\linewidth, keepaspectratio]{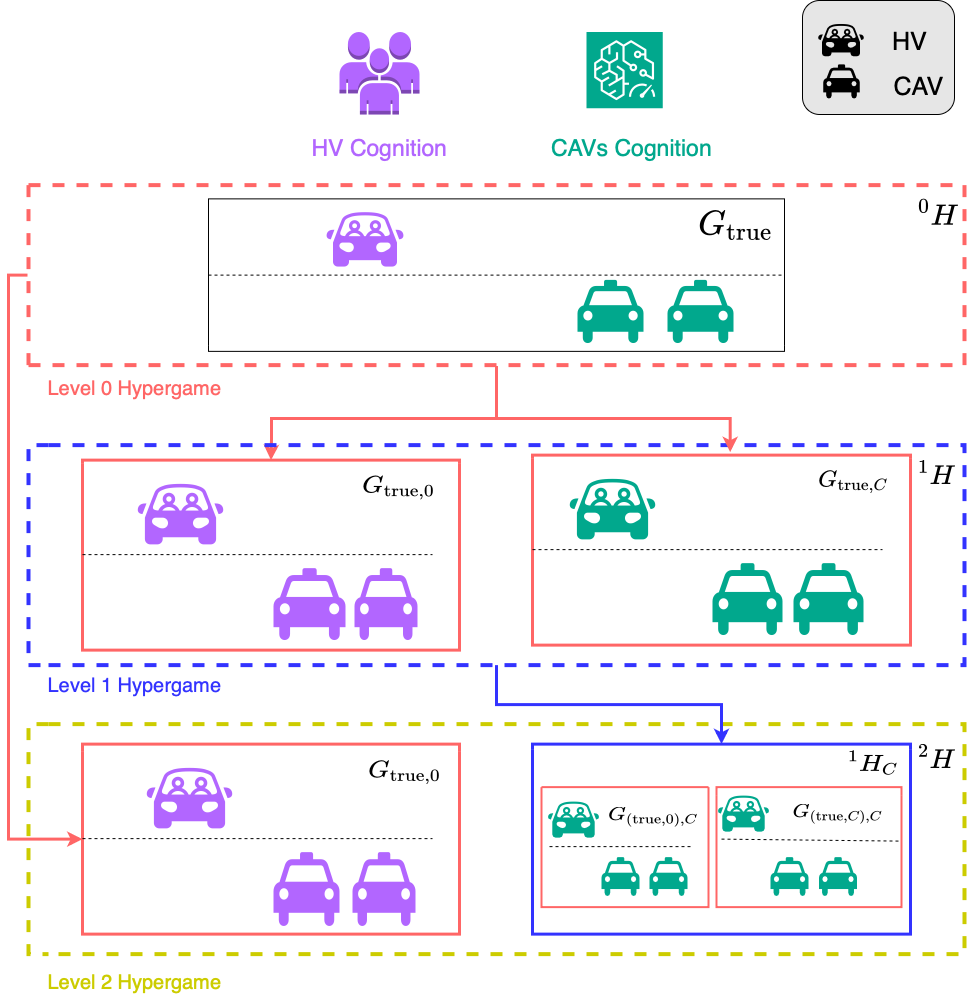}}
    \caption{The cognitive structure of the HV and CAVs in the hypergame. The dotted-line boxes represent different levels of hypergames: zero level, first level, and second level. On the left side, the solid line boxes indicate HV' cognition of the game at each level of the hypergame, while the right side represents CAVs' cognition, denoted by vehicles of the same color. The box pointed to by the arrow indicates the player's overall understanding of the game at a lower-level hypergame within the context of the current higher-level hypergame, as indicated by boxes and arrows of the same color.}
    \label{fig:cognition}
\end{figure}

\subsubsection{Level 0 and Level 1 Hypergames}  
For any \( i \in \mathcal{N} \), let \( G_{\text{true},i} \) represent vehicle \( i \)'s perception of \( G_{\text{true}} \), the actual game defined in Section \ref{sec:game_model}. To formalize parameter perception, define \( \theta_{j,i} \) as vehicle \( i \)'s estimation of \( \theta_j \), the parameter associated with vehicle \( j \), for all \( i, j \in \mathcal{N} \). Notably, \( \theta_{i,i} = \theta_{i,\text{true}} \), indicating that each vehicle \( i \in \mathcal{N} \) has perfect knowledge of its own parameter. Additionally, as to be explained in Remark~\ref{rem:same}, it follows that \( G_{\text{true},i} = G_{\text{true},j} \) and \( \theta_{l,i} = \theta_{l,j} \) for any \( i, j \in \mathcal{C} \) and \( l \in \mathcal{N} \).  

\begin{remark}\label{rem:same}  
    Since CAVs communicate seamlessly via V2X, their understanding of the game is assumed to be identical. Consequently, this work focuses primarily on the cognitive interplay between HV and the collective CAVs. For clarity, Figure~\ref{fig:cognition} consolidates the CAVs into a unified representation.  
\end{remark}  

In the red dashed box in Figure~\ref{fig:cognition}, the level 0 hypergame, denoted as \(\prescript{0}{}{H}\), represents the baseline game without cognitive discrepancies, defined as \(G_{\text{true}}\) in Game~\ref{de:level_0_hypergame}. While the level 1 hypergame accounts for the subjective perspectives of each player, where they perceive their own versions of the level 0 game but remain unaware of the perceptions held by others. Each player \( i \in \mathcal{N} \) interprets the game as \( G_{\text{true},i} \).  

As depicted in the blue dashed box in Figure~\ref{fig:cognition}, the level 1 hypergame is formalized as a tuple \(\prescript{1}{}{H} = \{G_{\text{true},i}, i \in \mathcal{N}\}\). Given the bounded rationality inherent in human cognition, the specific structure of \( G_{\text{true},0} \), representing the HV's perception of the game, is further elaborated in Game~\ref{de:hv_cognition}.  

\begin{game}
The game perceived by the HV, denoted as Game~\ref{de:level_0_hypergame}, is given by  
\[
G_{\text{true},0} = \left( \mathcal{N}, \{S_{i}(\mathbf{s}_{-i})\}_{i \in \mathcal{N}}, \{J_{i}(s_i; \theta_{i,0})\}_{i \in \mathcal{N}} \right),
\]
where the parameter \(\theta_{i,0}\) represents the HV’s understanding of the parameter \(\theta_{i,\text{true}}\) of vehicle \( i \). Specifically, \(\theta_{i,0} = \theta_{i,\text{ave}}\) for any \( i \in \mathcal{C} \), and \(\theta_{0,0} = \theta_{0,\text{true}}\).
\label{de:hv_cognition}
\end{game}

In the level 1 hypergame, the HV predicts the trajectories of CAVs and plans its own trajectory based on Game~\ref{de:hv_cognition}. The concept of a subjective rationalization strategy for the HV is formalized as follows.

\begin{definition}\label{def:hv_stra}
For the HV, a strategy \( s_{0}^* \) is said to be a subjective rationalization strategy if it forms part of a generalized Nash equilibrium (GNE) of \( G_{\text{true},0} \). This implies the existence of \(\{ s_{i,0}^* \}_{i \in \mathcal{C}}\) such that
\begin{equation*}
    \begin{aligned}
        J_{0}(s_{0}^*; \theta_{0, \text{true}}) \leq J_{0}(s_{0}; \theta_{0, \text{true}}), \quad \forall s_{0} \in S_{0}\left( \operatorname{vec}(s_{j,0}^*, j \in \mathcal{C}) \right), \\
        J_{i}(s_{i,0}^*; \theta_{i, \text{ave}}) \leq J_{i}(s_{i,0}; \theta_{i, \text{ave}}), \quad \forall s_{i,0} \in S_{i}\left( \operatorname{vec}(\mathbf{s}_{\neg i,0}^*, s_{0}^*) \right), \\
        \forall i \in \mathcal{C}.
    \end{aligned}
\end{equation*}
\end{definition}

Definition~\ref{def:hv_stra} signifies that within the HV's cognition, it perceives no benefit in unilaterally deviating from its chosen strategy \( s_{0}^* \), given its predictions of CAV behavior.

\subsubsection{Level 2 Hypergame}
In a level 2 hypergame, at least one player recognizes that different games are played due to the presence of misperceptions. In this study, we assume that CAVs are aware of these differing games, as they account for the cognition of HV.

Multiple superscripts are used to denote multiple levels of cognition. Each index represents the cognition of the entire variable to its left. For instance, $G_{(\text{true}, i), j}$ represents the second-order cognition of $\mathcal{G}$. Here, vehicle $i$ first forms an understanding of $\mathcal{G}_{\text{true}}$ as $\mathcal{G}_{\text{true}, i}$, and subsequently, vehicle $j$ develops an understanding of vehicle $i$'s cognition. Similarly, $\theta_{(i, j), l}$ represents the second-order cognition of vehicle $i$'s parameter $\theta_i$, where vehicle $j$ first perceives $\theta_i$ as $\theta_{i,j}$, and subsequently, vehicle $l$ understands vehicle $j$'s perception.

When CAVs are aware that HVs are playing a different game in a level 2 hypergame, CAV $j$'s perception of Game~\ref{de:hv_cognition} is given as follows:

\begin{game}\label{de:hv_cognition_in_cavs_cognition}
The CAV $j \in \mathcal{C}$'s perception of Game~\ref{de:hv_cognition} is
\[
G_{(\text{true},0),j} = \left( \mathcal{N}, \set{S_{i}(\mathbf{s}_{-i})}_{i\in \mathcal{N}}, \set{J_{i}(s_i;\theta_{(i,0),j})}_{i\in \mathcal{N}} \right),  
\]
where $\theta_{(i,0),j}$ represents CAV $j$'s understanding of $\theta_{i,0}$, which is HV's perception of the parameter $\theta_{i,\text{true}}$ of vehicle $i$. Specifically, $\theta_{(i,0),j} = \theta_{i,\text{ave}}$ and $\theta_{(0,0),j} = \theta_{0,j}$ for $i, j \in \mathcal{C}$.
\end{game}

According to Remark~\ref{rem:same}, all CAVs share the same perception of HVs, so $\theta_{0,i} = \theta_{0,j}$ for any $i, j \in \mathcal{C}$. We denote this shared perception as $\theta_{0,C}$. Furthermore, in CAVs' perception, HV's subjective rationalization strategy is consistent, denoted as $s_{0,C}$, implying that HV will not unilaterally deviate from this strategy. Based on $G_{(\text{true},j),i}$ for $j \in \mathcal{N}$, this leads to the subjective rationalization strategy for CAVs defined below:

\begin{definition}\label{def:cav_str}
For CAVs, a strategy profile $\set{s_{i}^*}_{i \in \mathcal{C}}$ is said to be a subjective rationalization strategy if there exists $s_{0,C}$, the subjective rationalization strategy of HV in Game~\ref{de:hv_cognition_in_cavs_cognition}, such that for any $i \in \mathcal{C}$:
    \begin{equation*}
        \begin{aligned}
            J_{i}(s_{i}^* ; \theta_{i, \text{true}}) \leq J_{i}(s_{i} ; \theta_{i, \text{true}}), \forall s_{i} \in S_{i}\left( \operatorname{vec}(\mathbf{s}_{\neg i}^*,s_{0,C}) \right).
        \end{aligned}
    \end{equation*}
\end{definition}

The subjective rationalization strategy for CAVs ensures that no CAV unilaterally changes its strategy in their perceived game. The level 1 hypergame, ${}^1H$, perceived by CAV $i \in \mathcal{C}$ is defined as ${}^1H_{i} = \{G_{(\text{true},j),i}, j \in \mathcal{N}\}$, where $G_{(\text{true},j),i}$ is as described in Game~\ref{de:hv_cognition_in_cavs_cognition}. The level 2 hypergame is then defined as follows:

\begin{game}\label{game:hpyer}
    The level 2 hypergame is a tuple \({}^2H = \{G_{\text{true},0}, {}^1H_{i}, i \in \mathcal{C} \}\), where $G_{\text{true},0}$ and ${}^1H_{i}$ are as defined above.
\end{game}

Game~\ref{game:hpyer} encapsulates the differing cognitive perspectives between HVs and CAVs in the level 2 hypergame context, assuming that each player acts rationally based on their own cognition. This leads to the concept of an Hyper Nash Equilibrium (HNE).

\begin{definition}\label{def:hne}
A strategy profile $\mathbf{s}^*$ is an HNE in the game ${}^2H$ if $\set{s_{i}^*}_{i \in \mathcal{C}}$ is the subjective rationalization strategy of CAVs defined in Definition~\ref{def:cav_str}, and $s_{0}^*$ is the subjective rationalization strategy of HV defined in Definition~\ref{def:hv_stra}, satisfying:
    \begin{equation*}
        \begin{aligned}
            &J_{i}(s_{i}^* ; \theta_{i,\text{true}}) \leq J_{i}(s_{i} ; \theta_{i,\text{true}}), \forall s_{i} \in S_{i}\left( \mathbf{s}_{-i}^* \right), \forall i \in \mathcal{C}, \\
            &J_{0}(s_{0}^* ; \theta_{0,\text{true}}) \leq J_{0}(s_{0} ; \theta_{0,\text{true}}) + \epsilon_p, \forall s_{0} \in S_{0}\left( \mathbf{s}_{-0}^* \right),
        \end{aligned}
    \end{equation*}
where $\epsilon_p$ is the perceptual threshold given in Section~\ref{sec:Perceptual}.
\end{definition}

In essence, an HNE is a strategy profile where each player is playing their best response within their respective subjective game, which is formed based on their perception of the overall situation. In this equilibrium, no player would unilaterally deviate from their current strategy, as doing so would not provide them with any additional benefit under their subjective understanding of the game. Furthermore, this equilibrium reflects a state of cognitive stability, as players do not have an incentive to alter their perception of the game itself. In other words, at an HNE, players not only achieve strategic stability by optimizing their actions but also maintain consistency in their mental models of the game. This dual stability ensures that players are aligned with their perceived realities, making the HNE a robust solution concept in hypergames \cite{Cheng2021SingleLeaderMultipleFollowersSS,xu2024consistency}

\section{Cognitive Stability Analysis}\label{sec:stability}

In this section, we consider a refined solution concept of GNE, namely the variational equilibrium. We establish the conditions under which the rationalizable strategies of the players constitute an HNE, assuming that CAVs have knowledge of the true objective function parameters of HV, which provides a cognitive stability analysis of the proposed model.

We first define the strategy profile excluding the strategy of HV as \(\mathbf{s} = \operatorname{vec}(s_i, \, i \in \mathcal{C})\), and the pseudo-gradient as  
\[
\mathcal{J}(\mathbf{s}; \theta) = 
\begin{bmatrix}
\nabla_{s_1} J_1(s_1; \theta_1) \\
\nabla_{s_2} J_2(s_2; \theta_2) \\
\vdots \\
\nabla_{s_n} J_n(s_n; \theta_n)
\end{bmatrix}.
\]
Specifically, the gradient of the cost function \( J_i(s_i; \theta_i) \) with respect to \( s_i \) is given by
\begin{equation}\label{eq:def_nabla_J}
    \nabla_{s_i} J_i(s_i; \theta_i) = \Bar{\theta_i}(s_i - s_{\text{ref},i}),
\end{equation}
where
\begin{equation}\label{eq:def_hat_theta}
    \Bar{\theta_i} = R \oplus \underbrace{Q \oplus R \oplus Q \oplus \cdots \oplus R}_{T-2 \text{ alternating } Q \text{ and } R} \oplus Q
\end{equation}
is a diagonal matrix of size \( 6(|T|-1) \times 6(|T|-1) \). Here, \( s_{\text{ref},i} \) denotes a reference trajectory vector aligned with \( s_i \), whose elements are defined as follows: the element corresponding to the state \( x_i \) in \( s_i \) is set to \( x_{\text{ref},i} \), while the elements corresponding to the control inputs \( u_i \) in \( s_i \) are set to zero.

Since we consider only the linear form of all constraints in Remark~\ref{re:line},
according to Lemma 2 of \cite{liu2024semi}, we know that given the strategy $s_0$ of HV, there exists a closed convex set \(K(s_0)\) such that for all \(i \in \mathcal{C}\),
\[
S_i(\mathbf{s}_{\neg i}) = \{ s_i \mid (s_i, \mathbf{s}_{\neg i}) \in K(s_0) \}.
\]

Given the strategy $s_0$ of HV and the parameter $\theta$ in cost functions, we define the strategy profile \(\mathbf{s}^* \in K(s_0)\) as a Variational Equilibrium (VE) if it satisfies the following variational inequality:  
\begin{equation}\label{eq:vepro}
    \langle \mathcal{J}(\mathbf{s}; \theta), \mathbf{s} - \mathbf{s}^* \rangle \geq 0, \quad \forall \mathbf{s} \in K(s_0).
\end{equation}
\noindent This condition guarantees that no player can improve their objective by unilaterally deviating from the strategy, ensuring the stability of the strategy profile.

\begin{remark}\label{rm:use_ve}
    According to Theorem 4.8 in \cite{facchinei2010generalized}, if \(\mathbf{s}^*\) is a VE satisfying \eqref{eq:vepro}, it is also a generalized Nash equilibrium (GNE). Furthermore, VE serves as a refinement of the GNE, making it a more preferred concept for equilibrium analysis \cite{kulkarni2012variational}. In the game-theoretical trajectory interaction solutions of vehicles, the VE is an interaction-fair GNE, meaning that both vehicles bear the same rate of payoff decrease to avoid collisions \cite{liu2024semi}. Therefore, we simplify the analysis of cognitive stability by focusing on the stability of the VE in this section. This approach enables a more precise understanding of cognitive stability in the context of the hypergame framework.
\end{remark}

As described in Remark~\ref{rm:use_ve}, we only use VE as the solution of the trajectory game in this section. The following theorem establishes a sufficient condition for achieving an HNE within the hypergame framework.

\begin{theorem}\label{thm:hne}
    Under the cognitive threshold $\epsilon_c$, if the CAVs can observe the true parameters of the HVs $\theta_{0,\text{true}}$, then the subjectively rationalized strategy profile $\{s_i^*\}_{i \in \mathcal{C}} \cup \{s_0^*\}$ of the CAVs and HV forms an HNE under the perceptual threshold $L\epsilon_c$, where $L$ is a positive constant.
\end{theorem}

\begin{proof}
    To prove the theorem, we first show that the function $\mathcal{J}(\mathbf{s}; \theta)$ is strongly monotone in $\mathbf{s}$ and Lipschitz continuous in both $\mathbf{s}$ and $\theta$. Define $\mathcal{\hat{J}}(\mathbf{\hat{s}}; \theta)=\mathcal{J}(\mathbf{s}+\mathbf{s}_{\text{ref}}; \theta)$, where $\mathbf{s}_{\text{ref}} = \operatorname{vec}(s_{\text{ref},i},i \in \mathcal{C})$. Then $\mathcal{\hat{J}}(\mathbf{\hat{s}}; \theta) = \Bar{\theta}\mathbf{\hat{s}}$, where $\Bar{\theta} = \oplus_{i \in \mathcal{C}}\Bar{\theta}_i$ is a diagonal matrix ($\Bar{\theta}_i$ is defined in \eqref{eq:def_hat_theta}). Therefore, $\mathbf{s}^*$ is the VE of \eqref{eq:vepro} if and only if $\mathbf{\hat{s}}^* = \mathbf{s}^* - \mathbf{s}_{\text{ref}}$ is the solution of the following variational inequality:
    \begin{equation}\label{eq:vepro2}
    \langle \mathcal{\hat{J}}(\mathbf{\hat{s}}; \theta), \mathbf{\hat{s}} - \mathbf{\hat{s}}^* \rangle \geq 0, \quad \forall \mathbf{\hat{s}} \in K(s_0) - \mathbf{s}_{\text{ref}},
    \end{equation}
    where $K(s_0) - \mathbf{s}_{\text{ref}} = \{ \mathbf{\hat{s}} - \mathbf{s}_{\text{ref}} \mid \mathbf{\hat{s}} \in K(s_0) \}$ is also a closed convex set.
    
    First, because there exists a lower bound  $\theta_{\min} > 0$ for every possible parameters in cost functions  as described in Subsection~\ref{subsec:obj_fun}, we have that for each $\theta \in \Theta$,
    \begin{equation*}
        \begin{aligned}
            &\langle \mathcal{\hat{J}}(\mathbf{\hat{s}}; \theta) - \mathcal{\hat{J}}(\mathbf{\hat{s}}'; \theta), \mathbf{\hat{s}} - \mathbf{\hat{s}}' \rangle \\
            = & \| \mathbf{\hat{s}} - \mathbf{\hat{s}}' \|_{\Bar{\mathbf{\theta}}}^2 \\
            \geq & \theta_{\min} \| \mathbf{\hat{s}} - \mathbf{\hat{s}}' \|^2, \quad \forall \mathbf{\hat{s}}, \mathbf{\hat{s}}' \in  K(s_0) - \mathbf{s}_{\text{ref}}.
        \end{aligned}
    \end{equation*}
    Thus, we obtain that $\mathcal{\hat{J}}(\mathbf{\hat{s}}; \theta)$ is strongly monotone with respect to $\mathbf{\hat{s}}$. Similarly, since there exists an upper bound $\theta_{\max} > 0$, we have 
    \[
        \langle \mathcal{\hat{J}}(\mathbf{\hat{s}}; \theta) - \mathcal{\hat{J}}(\mathbf{\hat{s}}'; \theta), \mathbf{\hat{s}} - \mathbf{\hat{s}}' \rangle \leq \theta_{\max} \| \mathbf{\hat{s}} - \mathbf{\hat{s}}' \|^2, \quad \forall \mathbf{\hat{s}}, \mathbf{\hat{s}}' \in  K(s_0) - \mathbf{s}_{\text{ref}}.
    \]
    This shows that $\mathcal{\hat{J}}(\mathbf{\hat{s}}; \theta)$ is Lipschitz continuous with respect to $\mathbf{\hat{s}}$. Moreover, for any $\theta, \theta' \in \Theta$, we have
    \begin{equation*}
        \begin{aligned}
             &\| \mathcal{\hat{J}}(\mathbf{\hat{s}}; \theta) - \mathcal{\hat{J}}(\mathbf{\hat{s}}; \theta') \|^2 \\
            =& \mathbf{\hat{s}}^{\top}(\Bar{\theta} - \Bar{\theta}')^2\mathbf{\hat{s}} \\
            \leq& \| \mathbf{\hat{s}} \|^2 \max_{i\in\mathcal{C},j\in[6]}(\theta_{i,j} - \theta'_{i,j})^2 \\
            \leq& M \| \mathbf{\hat{s}} \|^2 \| \Bar{\theta} - \Bar{\theta}' \|^2, \forall \mathbf{\hat{s}} \in  K(s_0) - \mathbf{s}_{\text{ref}},
        \end{aligned}
    \end{equation*}
    where $M$ is a positive constant. 
    Hence, $\mathcal{\hat{J}}(\mathbf{\hat{s}}; \theta)$ is Lipschitz continuous with respect to $\theta$. Then, according to Theorem 1 in \cite{maugeri2009global}, there exists a unique VE solution $\mathbf{s}^*(\theta) = \mathbf{\hat{s}}^*(\theta) + \mathbf{s}_{\text{ref}}$ of the variational inequality \eqref{eq:vepro} and the solution is $\gamma_1$-Lipschitz continuous in $\Theta$ where $\gamma_1$ is a positive constant.

    Since $\{s_i^*\}_{i \in \mathcal{C}}$ represents the CAVs' subjective rationalization strategy profile defined in Definition~\ref{def:cav_str}, there exists a strategy $s_{0,C}$ for the HV that satisfies
    \begin{equation*}
        \begin{aligned}
            J_i(s_i^* ; \theta_{i, \text{true}}) &\leq J_i(s_i ; \theta_{i, \text{true}}), \\
            & \forall s_i \in S_i\left( \operatorname{vec}(\mathbf{s}_{\neg i}^*, s_{0,C}) \right), 
            \forall i \in \mathcal{C}.
        \end{aligned}
    \end{equation*}
    Therefore, according to Remark~\ref{rm:use_ve}, $\{s_i^*\}_{i \in \mathcal{C}}$ also is the solution to the following variational inequality
    \begin{equation*}
        \langle \mathcal{J}(s; \{ \theta_{i, \text{true}} \}_{i \in \mathcal{C}}), s - s^* \rangle \geq 0, \quad \forall s \in K(s_{0,C}).
    \end{equation*}
    When the CAVs know the true parameters of the HVs, $\theta_{0,\text{true}}$, they accurately perceive the HV's strategy. In this case, Game~\ref{de:hv_cognition_in_cavs_cognition} is equivalent to Game~\ref{de:hv_cognition}, so we have
    \[
    s_{0,C} = s_0^*.
    \]
    Thus, $\{s_i^*\}_{i \in \mathcal{C}}$ also is the solution to the following variational inequality
    \begin{equation}\label{eq:vi_1}
        \langle \mathcal{J}(s; \{ \theta_{i, \text{true}} \}_{i \in \mathcal{C}}), s - s^* \rangle \geq 0, \quad \forall s \in K(s_0^*).
    \end{equation}
    Since $s_0^*$ is the HV's subjective rationalization strategy defined in Definition~\ref{def:hv_stra}, there exists a strategy profile $\{s_{i,0}\}_{i \in \mathcal{C}}$ of CAVs such that
    \begin{equation}\label{eq:hv-str}
        \begin{aligned}
            J_0(s_0^* ; \theta_{0, \text{true}}) &\leq J_0(s_0 ; \theta_{0, \text{true}}), \\
            &\forall s_0 \in S_0\left( \operatorname{vec}( s_{i,0}, i \in \mathcal{C} ) \right),
        \end{aligned}
    \end{equation}
    and
    \[
        \begin{aligned}
            J_i(s_{i,0} ; \theta_{i, \text{ave}}) &\leq J_i(s_i ; \theta_{i, \text{ave}}), \\ 
            &\forall s_i \in S_i\left( \operatorname{vec}( \mathbf{s}_{\neg i,0}, s_0^* ) \right), \forall i \in \mathcal{C}.
        \end{aligned}
    \]
    According to Remark~\ref{rm:use_ve}, $\{s_{i,0}\}_{i \in \mathcal{C}}$ satisfies the following variational inequality
    \begin{equation}\label{eq:vi_2}
        \langle \mathcal{J}(s; \{ \theta_{i, \text{ave}} \}_{i \in \mathcal{C}}), s - s^* \rangle \geq 0, \quad \forall s \in K(s_0^*).
    \end{equation}
    Recall the above proven result which says that the solution of the variational inequality problem \eqref{eq:vepro} is $\gamma_1$-Lipschitz continuous in $\Theta$. Since $\| \theta_{i, \text{true}} - \theta_{i, \text{ave}} \| \leq \epsilon_c$, combining \eqref{eq:vi_1} and \eqref{eq:vi_2}, we obtain
    \begin{equation}\label{eq:first_ineq}
        \begin{aligned}
            &\| \operatorname{vec}(s_{i,0}, i \in \mathcal{C}) - \operatorname{vec}(s_i^*, i \in \mathcal{C}) \|^2 \\
            \leq& \gamma_1 \| \operatorname{vec}(\theta_{i, \text{ave}}, i \in \mathcal{C}) - \operatorname{vec}(\theta_{i, \text{true}}, i \in \mathcal{C}) \|^2  \\
            \leq& n\gamma_1 \epsilon_c^2,
        \end{aligned}
    \end{equation}
    where $\gamma_1$ is a positive constant.

    From \eqref{eq:hv-str}, the HV’s subjective rationalization strategy $s_0^*$ satisfies
    \[
        J_0(s_0^* ; \theta_{0, \text{true}}) = \min\{ J_0(s_0 ; \theta_{0, \text{true}}) \mid S_0\left( \operatorname{vec}( s_{i,0}, i \in \mathcal{C} ) \right) \}.
    \]
    Therefore, according to Theorem 3.1 in \cite{dempe2015lipschitz}, which establishes the $\gamma_2$-Lipschitz continuity of the optimal value function, we have
    \[
        \begin{aligned}
            &J_0(s_0^* ; \theta_{0, \text{true}}) - \min\{ J_0(s_0 ; \theta_{0, \text{true}}) \mid S_0\left( \operatorname{vec}( s_{i}^*, i \in \mathcal{C} ) \right) \} \\
            \leq &\gamma_2 \| \operatorname{vec}(s_{i,0}, i \in \mathcal{C}) - \operatorname{vec}(s_i^*, i \in \mathcal{C}) \|,
        \end{aligned}
    \]
    where $\gamma_2$ is a positive constant. Moreover, combining the inequality \eqref{eq:first_ineq}, we obtain that
    \[
        \begin{aligned}
            &J_0(s_0^* ; \theta_{0, \text{true}}) - \min\{ J_0(s_0 ; \theta_{0, \text{true}}) \mid S_i\left( \operatorname{vec}( s_{j}^*, j \in \mathcal{C} ) \right) \} \\
            \leq& \gamma_1 \gamma_2 \sqrt{n} \epsilon_c.
        \end{aligned}
    \]
    Therefore, for the strategy profile $\{s_i^*\}_{i\in \mathcal{C}}\cup\{s_0^*\}$ where $\{s_i^*\}_{i\in \mathcal{C}}$ is the CAVs' subjective rationalization strategy profile and $\{s_0^*\}$ is the HV's subjective rationalization strategy, it satisfies
        \begin{equation*}
        \begin{aligned}
            &J_{i}(s_{i}^* ; \theta_{i,\text{true}}) \leq J_{i}(s_{i} ; \theta_{i,\text{true}}), \forall s_{i} \in S_{i}\left( \mathbf{s}_{-i}^* \right), \forall i \in \mathcal{C}, \\
            &J_{0}(s_{0}^* ; \theta_{0,\text{true}}) \leq J_{0}(s_{0} ; \theta_{0,\text{true}}) + L\epsilon_c, \forall s_{0} \in S_{0}\left( \mathbf{s}_{-0}^* \right),
        \end{aligned}
    \end{equation*}
    where $L = \gamma_1 \gamma_2 \sqrt{n}$. By recalling the definition of HNE in Definition~\ref{def:hne}, we obtain that the strategy profile is an HNE under the cognitive threshold $\epsilon_c$ and perceptual threshold $L\epsilon_c$.
\end{proof}

Theorem~\ref{thm:hne} provides a detailed analysis of cognitive stability in the HNE achieved when CAVs successfully learn the parameters of HV. This result underscores the critical role of accurate parameter estimation in ensuring cognitive stability, as it allows CAVs to align their strategies with the actual driving behavior and preferences of HV. By understanding the underlying objectives and constraints of HV, CAVs can anticipate their actions effectively, reducing the potential for conflicts and misunderstandings in mixed traffic environments. 

The following section delves into the methods through which CAVs acquire this knowledge, namely, inverse learning based on observed game trajectories. This process involves leveraging data from past interactions to infer the parameters governing HV's decision-making models. By identifying these parameters, CAVs can reconstruct the subjective games played by HVs and adapt their own strategies accordingly. This capability enables CAVs to proactively plan their actions in a manner that promotes harmony and efficiency in traffic dynamics, thereby contributing to the overall safety and performance of the system.

\section{Inverse Learning-Based Intention Interpretation and Distributed Trajectory Planning}
\label{sec:methods}
In this section, we explore intention recognition and distributed trajectory planning within the multi-level hypergame cognitive framework, distinguishing between offline and online scenarios and utilizing inverse learning techniques. We use the lane-change scenarios commonly used in autonomous driving \cite{10334023}.

We first present the algorithm \texttt{SolveGames}, as shown in Algorithm \ref{alg:solveGame}, which will be used in the subsequent algorithms. \texttt{SolveGames} is a general method for CAVs to solve game problems defined in this paper. 
Due to the generality of Algorithm \ref{alg:solveGame}, the specific meaning of its input and output varies with the problem, so we use $\tilde{(\cdot)}$ to denote general symbols to distinguish them from the notation above. For example, $\tilde{\mathcal{N}}$ can be $\mathcal{N}$ or $\mathcal{C}$, and $\tilde{\theta}_i$ can be $\theta_{i,\text{true}}$ or $\theta_{(i,0),i}$.
Given the parameter  $\tilde{\theta}_i$ of each player $i \in \tilde{\mathcal{N}}$ in the game, CAVs and the RSU collaboratively and distributedly compute the generalized Nash equilibrium $\tilde{\mathbf{s}}$ based on Algorithm \ref{alg:solveGame}. The index $\zeta$ indicates the iteration count. We choose the relative step progress and constraint violation threshold as the stopping criterion \cite{boyd2004convexO}, which is computed and judged by the RSU.
By default, we use reference trajectories to generate the input $\tilde{s}_i^0$ of Algorithm \ref{alg:solveGame}, thus $\tilde{s}_i^0$ is omitted in subsequent calls to Algorithm \ref{alg:solveGame}.

\begin{algorithm}
    \caption{\texttt{SolveGames}}
    \begin{algorithmic}[1] \label{alg:solveGame}
        \REQUIRE{$\tilde{\theta}_i, \tilde{s}_i^0, i \in \tilde{\mathcal{N}}$, maximum iteration times $\zeta _{\max}$}. 
        \ENSURE{Strategy profile $\tilde{\mathbf{s}}$}
        \FOR{$\zeta =0:\zeta _{\max}$}
        \FOR{$i \in \tilde{\mathcal{N}}$}
        \IF{$i==0$}
        \STATE {\it Communication:} The RSU receives  $\tilde{s}_{-0}^{\zeta}$ from CAVs;
        \STATE {\it RSU:}
        \[
        \tilde{s}_0^{\zeta+1} \gets \operatorname{argmin}\{ J_0(\tilde{s}_0;\tilde{\theta}_{0}) \mid \tilde{s}_0\in S_0(\tilde{\mathbf{s}}_{-0}^{\zeta}) \};
        \]
        \STATE {\it Communication:} The RSU sends $\tilde{s}_0^{\zeta+1}$ to CAVs;
        \ELSE
        \STATE {\it Communication:} CAV $i$ receives $\tilde{\mathbf{s}}_{-i}^{\zeta}$ from the RSU and other CAVs;  
        \STATE {\it CAV $i$:}
        \[
        \tilde{s}_i^{\zeta+1} \gets \operatorname{argmin}\{J_i(\tilde{s}_i;\tilde{\theta}_{i})\mid \tilde{s}_i\in S_i(\tilde{\mathbf{s}}_{-i}^{\zeta})\};
        \]
        \STATE {\it Communication:} CAV $i$ sends $\tilde{s}_i^{\zeta+1}$ to the RSU and other CAVs;
        \ENDIF
        \ENDFOR
        \IF{{\it RSU:} stopping criterion are met}
        \STATE \textbf{break};
        \ENDIF
        \ENDFOR
    \end{algorithmic}
\end{algorithm}

We divide the entire interaction process of vehicles on the lane into discrete times of $T$. The following introduces intent recognition and trajectory planning for CAVs in offline and online scenarios respectively.

\subsection{Offline Scenario}
In the offline scenario, the entire interaction process between vehicles is considered as a game, that is, $\mathcal{T}=\{1,2,\dots,T\}$. CAVs first recognize the intention of HVs through offline inverse learning, and then predict and plan their own trajectories.

\subsubsection{Intention Interpretation of HV by CAVs}
As evident from cognitive stability analysis in Section~\ref{sec:stability}, the accuracy of CAVs' perception of HV's weights $\theta_{0,C}$ is crucial for CAVs to achieve HNE and accurately predict HV's trajectory. 
CAVs cannot directly access HV's weights $\theta_{0,\text{true}}$. Therefore, they need to learn these from historical trajectories. This process of learning parameters from equilibrium or optimal solution is referred to as intention interpretation, which is in fact the inverse of Game~\ref{de:hv_cognition_in_cavs_cognition}. The following introduces how CAVs use the KKT-based inverse learning method to get its estimate of the HV parameter's $\theta_{0,C}$ \cite{Chen2023OnlinePI}.

When CAVs have the perfect perception of HV, namely $\theta_{0,C}=\theta_{0,\text{true}}$, Game~\ref{de:hv_cognition} and Game~\ref{de:hv_cognition_in_cavs_cognition} are identical. Therefore, the equilibrium $s_0$ and $s_{i,0}, \forall i \in \mathcal{C}$ from Game~\ref{de:hv_cognition} can be regarded as the ground truth states of $s_{0,C}$ and $s_{(i,0),i}$ from Game~\ref{de:hv_cognition_in_cavs_cognition}, respectively.
We assume that CAVs can observe the trajectory of HV, denoted as \(\hat{s}_0\), which may be a noise-perturbed version of the true trajectory \(s_0\).
Therefore, the intention interpretation problem is defined as Problem~\ref{de:intention_interpretation_problem}.

\begin{problem}
The intention interpretation problem for CAVs regarding the HV is the inverse of Game~\ref{de:hv_cognition_in_cavs_cognition}. The purpose is to get $\theta_{0,C}$ by observing HV's trajectory $\hat{s}_0$.
\label{de:intention_interpretation_problem}
\end{problem}

Specifically, CAVs collaboratively compute $\{s_{(i,0),C} \}_{i \in \mathcal{C}} $, which is the equilibrium strategy of CAVs perceived by HV in CAVs' understanding, using \texttt{SolveGames}($\{ \theta_{i, \text{ave}} \}_{i \in \mathcal{C}} $) in Algorithm \ref{alg:solveGame} by fixing HV's strategy $s_{0,C}$ as $\hat{s}_0$. Therefore, HV's decision model in CAVs' cognition is
\begin{equation} \label{eq:hv_decision_in_cavs_cognition}
    \hat{s}_0 = \argmin \{ J_0(s_0;\theta_{0,C}) \mid s_0 \in S_0(\{s_{(i,0),C} \}_{i \in \mathcal{C}} ) \} + \xi,
\end{equation}
where $\theta_{0,C}$ is HV's weights in CAVs' cognition and $\xi$ is an unknown random noise. By recalling the definition of the constraint set $S_0$ in \eqref{eq:cons}, we get that the KKT conditions of \eqref{eq:hv_decision_in_cavs_cognition} are 
\begin{equation} \label{eq:KKT_conditions}
\left\{
    \begin{array}{l}
      \nabla J_0(\hat{s}_0;\theta_{0,C}) + \nabla g_{0}(\hat{s}_0;\mathbf{s}_{(C,0),C})^{\top}\lambda +  \nabla h_{0}(\hat{s}_0;\mathbf{s}_{(C,0),C})^{\top} \mu =\mathbf{0},  \\
       \lambda \circ g_{0}(\hat{s}_0;\mathbf{s}_{(C,0),C}) =\mathbf{0}, \\
       \lambda \geq \mathbf{0}, \\
       h_{0}(\hat{s}_{0},\mathbf{s}_{(C,0),C}) = \boldsymbol{0}, \\
       g_{0}(\hat{s}_{0},\mathbf{s}_{(C,0),C}) \leq \boldsymbol{0}.
    \end{array}
\right.
\end{equation}
Based on the KKT conditions in \eqref{eq:KKT_conditions} without noises, CAVs can get $\theta_{0,C}$ by solving the following optimization:
\begin{equation} \label{eq:update_opt}
    \begin{split}
        \min\limits_{\theta,\lambda,\mu} & \lVert \nabla J_0(\hat{s}_0;\theta) + \nabla g_{0}(\hat{s}_0;\mathbf{s}_{(C,0),C})^{\top}\lambda +  \nabla h_{0}(\hat{s}_0;\mathbf{s}_{(C,0),C})^{\top} \mu \rVert_2^2, \\
        \mathrm{s.t.} & \lambda \geq \mathbf{0}, \theta \in \Theta,  \\
        &\lambda \circ \min \set{g_{0}(\hat{s}_0;\mathbf{s}_{(C,0),C})+\kappa,\mathbf{0}}=\mathbf{0}.
    \end{split}
\end{equation}
where $\kappa>0$ is a small threshold to handle observation errors. 

We then summarize the above process into the following Algorithm \ref{alg:inverseGame}.

\begin{algorithm}
    \caption{\texttt{IntentionInterpretationOffline}}
    \begin{algorithmic}[1] \label{alg:inverseGame}
        \REQUIRE{HV's trajectory observed by CAVs $\hat{s}_0$}
        \ENSURE{CAVs' cognition $\theta_{0,C}$}
        \STATE Solve Game~\ref{de:hv_cognition_in_cavs_cognition} with fixed $s_{0,C}=\hat{s}_0$: 
        \[
        \{s_{(i,0),C} \}_{i \in \mathcal{C}} \gets \texttt{SolveGames}(\{ \theta_{i, \text{ave}} \}_{i \in \mathcal{C}},\mathcal{C});
        \]
        \STATE $\theta_{0,C} \gets$ Solving optimization problem \eqref{eq:update_opt};
    \end{algorithmic}
\end{algorithm}

\subsubsection{Trajectory Prediction and Planning Method of CAVs}
In this part, we will use the learned intentions to predict HV's trajectory and plan CAVs' trajectory during the actual process.
In the level 2 hypergame, the CAVs consider their perception of HV's decision model, Game~\ref{de:hv_cognition_in_cavs_cognition}, which is used to predict HV's trajectory $s_{0,C}$. Then the CAVs' decision model is given by Problem \ref{de:cavs_cognition}, where the CAVs' perception of themselves is accurate. Therefore, the parameters related to the CAVs in the game are the same as in Game~\ref{de:level_0_hypergame}, while HV's trajectory is fixed as the predicted trajectory $s_{0,C}$ obtained from Game~\ref{de:hv_cognition_in_cavs_cognition}.

\begin{problem}
The trajectory planning game of CAVs is defined as 
\begin{equation}
   s_i = \argminlimits\{ J_{i}\left(s_{i};\theta_{i,\text{true}}\right) \mid s_{i}\in S_i(s_{0,C},\mathbf{s}_{\neg i}) \}, i \in \mathcal{C}.
    \label{eq:cav_decision_problem}
\end{equation}
\label{de:cavs_cognition}
\end{problem}

In summary, the above process can be described as Algorithm~\ref{alg:process}.
\begin{algorithm}
    \caption{\textbf{Predicting and Planning under Different Cog-nition}}
    \begin{algorithmic}[1] \label{alg:process}
        \STATE $\theta_{0,C} \gets$ \texttt{IntentionInterpretationOffline}($\hat{s}_0$);
        \STATE Solve Game~\ref{de:hv_cognition_in_cavs_cognition}:
        \[
        s_{0,C} \gets \texttt{SolveGames}( \{ \theta_{i, \text{ave}} \}_{i \in \mathcal{C}} \cup \{\theta_{0,C} \}, \mathcal{N});
        \]
        \STATE Solve Problem \ref{de:cavs_cognition}:
        \[
        s_i \gets \texttt{SolveGames}(\theta_{i,\text{true}}, \mathcal{C}), i\in\mathcal{C}.
        \]
    \end{algorithmic}
\end{algorithm}

\subsection{Online Scenario}
When encountering a newcome HV, there is no offline data available for intent recognition. In this case, online intent recognition is required. In the following, we consider a multi-stage trajectory planning framework for vehicles within a prediction horizon \( T \).

The time horizon \( \{ 1, 2, \dots, T \} \) is divided into \( \tau > 0 \) sequential segments:
\[
\bigcup_{t=1}^{\tau} \mathcal{T}_t = \{ 1, 2, \dots, T \},
\]
where each subset \( \mathcal{T}_t \) represents a time segment: \( \mathcal{T}_t = \{k_{t-1}, \dots, k_{t}\} \), with \( 1 = k_0 < k_1 < \cdots < k_{\tau} = T \). At each time period \( \mathcal{T}_t \), we use a superscript \( ^t \) to indicate the corresponding games and variables, such as \( G^t \). Thus, the entire trajectory planning problem is modeled as a multi-stage online dynamic game, as illustrated in Figure~\ref{fig:online_scenario}.

\begin{figure}[htbp]
    \centerline{\includegraphics[width=\linewidth, keepaspectratio]{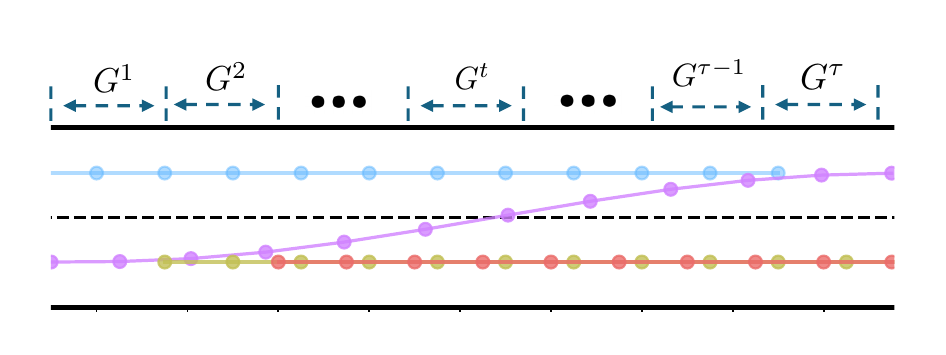}}
    \caption{Illustration of the online scenario with multi-stage trajectory games.}
    \label{fig:online_scenario}
\end{figure}

In the \( t \)-th game \( G^t \), the strategy of vehicle \( i \), denoted \( s_i^t \), is expressed as \( \operatorname{vec}(s_i(k), k \in \mathcal{T}_t) \), excluding the initial state \( x_i(k_{t-1}) \) and the terminal control \( u_i(k_{t}) \). The CAVs' estimate of the HV's true parameter \( \theta_{0,\text{true}} \) at time period \( t \) is denoted as \( \theta_{0,C}^t \).

\subsubsection{Intention Interpretation of HV by CAVs}
At time period \( t \geq 2 \), CAVs observe the HV's trajectory \( \hat{s}_0^{t-1} \) from the previous time period. Specifically, \( \hat{s}_0^{t-1} \) represents the equilibrium strategy \( s_0^{t-1} \) of the HV in \( G_{\text{true},0}^{t-1} \) (Game~\ref{de:hv_cognition}), perturbed by observational noise \( \xi^{t-1} \):
\[
\hat{s}_0^{t-1} = s_0^{t-1} + \xi^{t-1}.
\]
In this game, \( s_0^{t-1} \) satisfies the following conditions:
\begin{equation} \label{eq:1234}
    \begin{aligned}
    J_{0}^{t-1}(s_{0}^{t-1} ; \theta_{0, \text{true}}) &\leq J_{0}^{t-1}(s_{0} ; \theta_{0, \text{true}}), \\
    &\forall s_{0} \in S_{0}^{t-1}\left( \operatorname{vec}( s_{j,0}^{t-1}, j \in \mathcal{C} ) \right),
    \end{aligned}
\end{equation}
\begin{equation} \label{eq:1234fixed}
    \begin{aligned}
        J_{i}^{t-1}(s_{i,0}^{t-1} ; &\theta_{i, \text{ave}}) \leq J_{i}^{t-1}(s_{i,0} ; \theta_{i, \text{ave}}), \\
        &\forall s_{i,0} \in S_{i}^{t-1}\left( \operatorname{vec}( \mathbf{s}_{\neg i,0}^{t-1}, s_{0}^{t-1}  ) \right), \forall i \in \mathcal{C}.
    \end{aligned}
\end{equation}

Given \( \hat{s}_0^{t-1} \), CAVs calculate \( s_{(C,0), C}^{t-1} \) using their distributed computational capabilities and V2X communication. Specifically, they utilize the \texttt{SolveGame} algorithm to solve \eqref{eq:1234fixed}. To refine their cognition of the HV, CAVs update their estimate \( \theta_{0,C}^t \) by solving the following optimization:
\begin{equation} \label{eq:updata_opt}
    \begin{aligned}
        \min_{\theta,\lambda,\mu}& \underbrace{
        \begin{aligned}
            \| \nabla J_0(\hat{s}_0^{t-1};\theta) + \nabla g_{0}(\hat{s}_0^{t-1};\mathbf{s}_{(C,0),C}^{t-1}&)^{\top}\lambda \\
            +  \nabla h_{0}(&\hat{s}_0^{t-1};\mathbf{s}_{(C,0),C}^{t-1})^{\top} \mu \|_2^2
        \end{aligned}
        }_{\mathrm{correctiveness}} \\
        &+ \omega_{dist} \underbrace{\lVert \theta - \theta_{0,C}^{t-1} \rVert_2^2}_{\mathrm{conservativeness}} \\
        &\text{subject~to} \quad \lambda \geq 0, \theta \in \Theta  \\
        &\hspace{1.7cm} \lambda \circ \min \set{g_{0}(\hat{s}_0^{t-1};\mathbf{s}_{(C,0),C}^{t-1})+\kappa,\mathbf{0}}=\mathbf{0},
    \end{aligned}
\end{equation}
where \( \omega_{dist} \geq 0 \) is a weighting factor balancing `correctiveness' and `conservativeness'. The first term in \eqref{eq:updata_opt} ensures that the estimate aligns with observed HV behavior by minimizing deviations from the KKT conditions of \eqref{eq:1234}. The second term penalizes large deviations from the previous estimate, ensuring stability in updates. The parameter \( \omega_{dist} \) controls the trade-off between these competing objectives. The complete intention interpretation process is given in Algorithm~\ref{alg:inverseGame2}.
\begin{algorithm}
    \caption{\texttt{IntentionInterpretationOnline}}
    \begin{algorithmic}[1] \label{alg:inverseGame2}
        \REQUIRE{Cognition $\theta_{0,C}^{t-1}$ at time period $t-1$}
        \ENSURE{New cognition $\theta_{0,C}^{t}$}
        \STATE CAVs observe the trajectory $\hat{s}_0^t$ last stage;
        \STATE $\mathbf{s}_{(C,0),C}^t \gets \texttt{SolveGames}$ ($\{\theta_{i,\text{ave}}\}_{i \in \mathcal{C}}$) by fixing $s_{0,C}^t=\hat{s}_0^t$;
        \STATE $\theta_{0,C}^{t} \gets$ \eqref{eq:updata_opt}.
    \end{algorithmic}
\end{algorithm}

\subsubsection{Trajectory Prediction and Planning Method of CAVs}
After the intention interpretation process, CAVs utilize the learned intentions to predict the HV's trajectory and plan their own trajectories within the time period $\mathcal{T}_t$, similar to the offline scenario. Specifically, the CAVs incorporate their perception of the HV's decision model, defined as Game~\ref{de:hv_cognition_in_cavs_cognition}, to predict the HV's trajectory $s_{0,C}^t$. 

The trajectory prediction is then used as input for the CAVs' trajectory planning process. The decision-making problem for a CAV is formulated as:
\begin{equation}
    s_i^t = \argmin\{J_{i}(s_{i}; \theta_{i, \text{true}}) \mid s_i \in S_i^t(s_{0,C}^t, \mathbf{s}_{\neg i}^t) \}, \quad i \in \mathcal{C},
    \label{eq:cav_decision_problem_t}
\end{equation}
where the set of feasible strategies \(S_i^t\) considers the influence of the predicted HV trajectory $s_{0,C}^t$ and the strategies of other vehicles $\mathbf{s}_{\neg i}^t$. By leveraging V2X communication, CAVs can collaboratively solve this optimization problem in a distributed manner.

The entire online process is summarized in Algorithm~\ref{alg:process_t}.
\begin{algorithm}
    \caption{\textbf{Online Process}}
    \begin{algorithmic}[1] \label{alg:process_t}
        \STATE Initialize HV's parameter in CAVs' cognition $\theta_{0,C}^1$;
        \FOR{$t = 1:\tau$}
            \IF{$t > 1$}
                \STATE Update CAVs' cognition:
                \[
                \theta_{0,C}^t \gets \texttt{IntentionInterpretationOnline};
                \]
            \ENDIF
            \STATE Predict HV's trajectory by solving Game~\ref{de:hv_cognition_in_cavs_cognition}:
            \[
            s_{0,C}^t \gets \texttt{SolveGames}\big(\{\theta_{i, \text{ave}} \}_{i \in \mathcal{C}} \cup \{\theta_{0,C}^t\}\big);
            \]
            \STATE Plan CAVs' trajectories by solving \eqref{eq:cav_decision_problem_t}:
            \[
            s_i^t \gets \texttt{SolveGames}(\theta_{i, \text{true}}), \quad i \in \mathcal{C}.
            \]
        \ENDFOR
    \end{algorithmic}
\end{algorithm}

\section{Experimental Results}
In this section, we examine the performance of CAVs in recognizing, predicting, and interacting with HV during lane-changing scenarios in mixed traffic. Experiments are conducted under both offline and online conditions to ensure a comprehensive evaluation.

\subsection{Experimental Setting}
We evaluate and validate the algorithm's performance using a lane-changing task on a unidirectional, two-lane highway. Fig. \ref{fig:experimental_scenario} shows exemplary reference trajectories for each vehicle, with one HV traveling in the left lane and three CAVs traveling in the right lane. CAV $1$ plans to change lanes to the left, while the other vehicles plan to travel at a constant speed.

\begin{figure}[htbp]
    \centerline{\includegraphics[width=\linewidth, keepaspectratio]{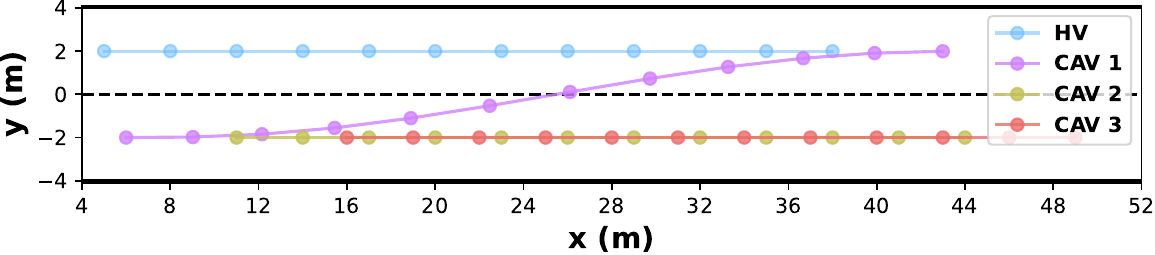}}
    \caption{The experimental scenario and reference trajectories for each vehicle.}
    \label{fig:experimental_scenario}
\end{figure}

In the experiment, the driving styles are classified into three types based on the norms of the components in $\theta$:
\begin{itemize}
    \item Pose-tracking: $\|(\theta_{p_x},\theta_{p_y},\theta_\psi)\|_2$ is the largest, indicating that the vehicle tends to track the positions and heading angles, i.e., the reference poses, in the reference trajectory.
    \item Velocity-consistent: $\theta_v$ is the largest, indicating that the vehicle tends to travel at the reference speed.
    \item Comfort-oriented: $\|(\theta_a,\theta_\delta)\|_2$ is the largest, indicating that the vehicle tends to use smaller control inputs, reflecting a preference for comfort.
\end{itemize}

\begin{table}[htbp]
    \caption{Typical ratios of weights for each driving style type.}
    \begin{center}
    \begin{tabular}{ccc}
        \toprule
        \shortstack{\textbf{Driving}\\ \textbf{Behavior}} & \shortstack{\textbf{Driving}\\\textbf{Style Type}} & \raisebox{1ex}{$\theta_{eff}$} \\
        \midrule
        \multirow{4}{*}{\shortstack{Straight-\\driving}} & Pose-tracking & $(10,1,1)^\top$ \\
        \cmidrule[\cmidruleWidth]{2-3}
                              & Velocity-consistent & $(1,10,1)^\top$ \\
        \cmidrule[\cmidruleWidth]{2-3}
                              & Comfort-oriented & $(1,1,10)^\top$ \\
        \midrule
        \multirow{4}{*}{\shortstack{Lane-\\changing}} & Pose-tracking & $(10,10,1,10,1,1)^\top$ \\
        \cmidrule[\cmidruleWidth]{2-3}
                              & Velocity-consistent & $(1,1,10,1,1,1)^\top$ \\
        \cmidrule[\cmidruleWidth]{2-3}
                              & Comfort-oriented & $(1,1,1,1,1,10)^\top$ \\
        \bottomrule
    \end{tabular}
    \label{tab:style_weight}
    \end{center}
\end{table}

The components of $\theta$ correspond one-to-one with the components of $s_i(k)$. The meanings of each component can be found in the definition of dynamics constraints in Sec. \ref{sec:model_of_game}. For straight-driving and lane-changing vehicles, the typical ratios of weights for each driving style type are shown in Table \ref{tab:style_weight}. In this scenario, straight-driving vehicles' driving behavior constraints cause them to travel along the horizontal line, so their effective weights are $\theta_{eff}=(\theta_{p_x}, \theta_{v}, \theta_{a})^\top$. For lane-changing vehicles, all weights are effective, i.e., $\theta_{eff}=\theta$. We always normalize parameters by $\frac{\theta_{eff}}{\|\theta_{eff}\|_2}$. The parameter settings in simulations are shown in Table \ref{tab:simulation_parameter}.

\begin{table}[htbp]
    \caption{Simulation parameters.}
    \begin{center}
        \begin{tabular}{cccc}
            \toprule
            \textbf{Parameter} & \textbf{Value} & \textbf{Parameter} & \textbf{Value} \\
            \midrule
            \shortstack{Vehicle size\\$L, W$} & \shortstack{$\qty{3.63}{m}$,\\$\qty{1.85}{m}$} & \shortstack{Extended vehicle\\size $L_E, W_E$} & \shortstack{$\qty{3.73}{m}$,\\$\qty{1.95}{m}$} \\
            \cmidrule[\cmidruleWidth]{1-2} \cmidrule[\cmidruleWidth]{3-4}
            \shortstack{Lane width} & $\qty{4}{m}$ & \shortstack{Range of $v$} & $[0, 20]\unit{m/s}$ \\
            \cmidrule[\cmidruleWidth]{1-2} \cmidrule[\cmidruleWidth]{3-4}
            \shortstack{Range of $a$} & $[-8, 2]\unit{m/s^2}$ & \shortstack{Range of $\delta$} & $[-33, 33]\unit{\degree}$ \\
            \cmidrule[\cmidruleWidth]{1-2} \cmidrule[\cmidruleWidth]{3-4}
            \shortstack{Constraint violation\\threshold $\epsilon$} & $\num{1e-3}$ & \shortstack{Discrete\\period $T_s$} & $\qty{0.1}{s}$ \\
            \cmidrule[\cmidruleWidth]{1-2} \cmidrule[\cmidruleWidth]{3-4}
            \shortstack{Relative step\\progress $\epsilon_{\mathrm{step}}$} & $\num{1e-2}$ & \shortstack{Maximum number \\
            of iterations $\zeta_{\max}$} & $\num{100}$ \\
            \bottomrule
        \end{tabular}
        \label{tab:simulation_parameter}
    \end{center}
\end{table}

\subsection{Offline Experiments}
In experiments, we measure the performance of the proposed method from the trajectory of the complete interaction process. Set $T=36$ and set the initial speed of each vehicle as $\qty{10}{m/s}$. The driving styles of the HV and CAVs $1$-$3$ are comfort-oriented, comfort-oriented, velocity-consistent, and pose-tracking, respectively. The observed HV's trajectory $\hat{s}_0$ is generated by adding Gaussian noise with a mean of $0$ to all $p_{x,0}(k)$ in $s_0$. The standard deviation of the Gaussian noise varies from $0.01$ to $0.40$ in increments of $0.01$, with each value tested $50$ times. The position observation error for the HV, defined as $\frac{1}{T}\|\hat{p}_{0}-p_{0}\|_2$, is used as a measure of the noise level, where $p_0=\operatorname{vec}(p_{x,0}\left( k \right) ,p_{y,0}\left( k \right) ),\forall k=2,\dots ,T$ is the position vector. The position observation error represents the average positional error between the observed trajectory and the actual trajectory at each time step. The algorithm's accuracy in learning HV's weights is evaluated using the parameter estimation error $\frac{\|\theta_{eff,0,C}-\theta_{eff,0}\|_2}{\|\theta_{eff,0}\|_2}$, which is the relative error between HV's weights in CAVs' cognition and  HV's actual weights. 

We make CAVs re-predict and re-plan trajectories at the initial moment using the learned parameters. The trajectory prediction error $\frac{1}{T}\|s_{0,C}-s_{0}\|_2$ is defined to measure the accuracy of trajectory prediction, and the position prediction error at each time step $\|p_{0,C}(k)-p_{0}(k)\|_2$ is defined to measure the accuracy of the position prediction in the trajectory. We set a relatively loose $\kappa=1.5$ to avoid misjudgment of the complementary slackness condition in the KKT conditions due to observation noise.

\begin{figure}[htbp]
    \centerline{\includegraphics[width=\linewidth, keepaspectratio]{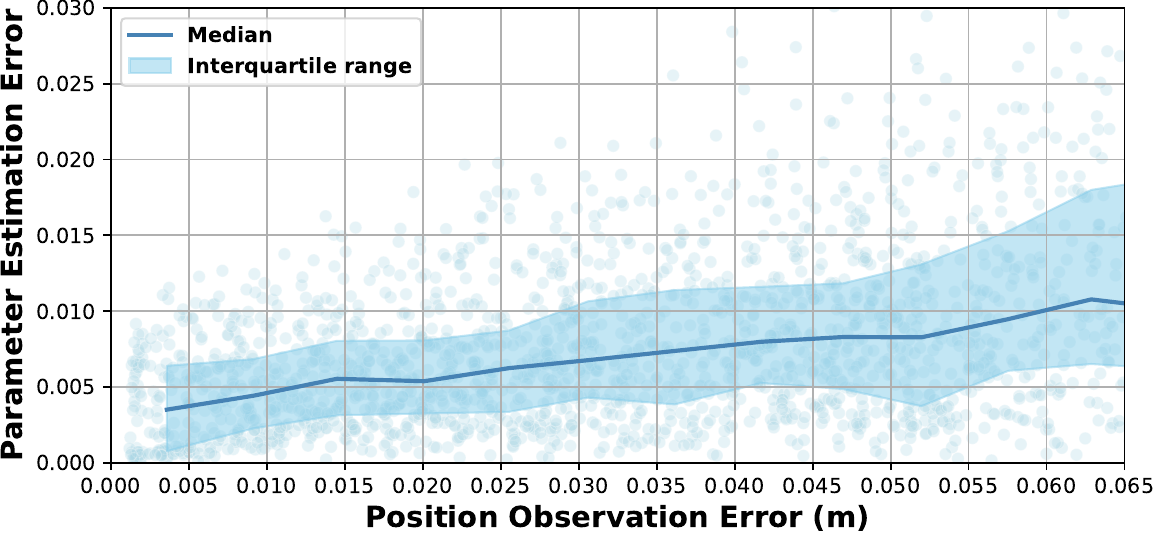}}
    \caption{The distribution of CAVs' parameter estimation errors for the HV under different position observation errors.}
    \label{fig:offline_parameter_position}
\end{figure}

Fig. \ref{fig:offline_parameter_position} presents the variation of parameter estimation errors with position observation errors, where the original data is represented as a scatter plot, the median is indicated by a line, and the interquartile range is visualized by the shaded area between the third and first quartiles. Fig. \ref{fig:offline_trajectory_position} shows how trajectory prediction errors based on learned HV's weights change with position observation errors. It can be seen that the accuracy of the weights learned by the algorithm remains high under the position observation noise, with only a slight decrease as the noise increases. Meanwhile, the trajectory prediction errors are significantly lower than the position observation errors. It is worth mentioning that the trajectory prediction errors include state and control errors, not just position errors, so these results suggest that the proposed method is robust at the trajectory prediction level.


\begin{figure}[htbp]
    \centerline{\includegraphics[width=\linewidth, keepaspectratio]{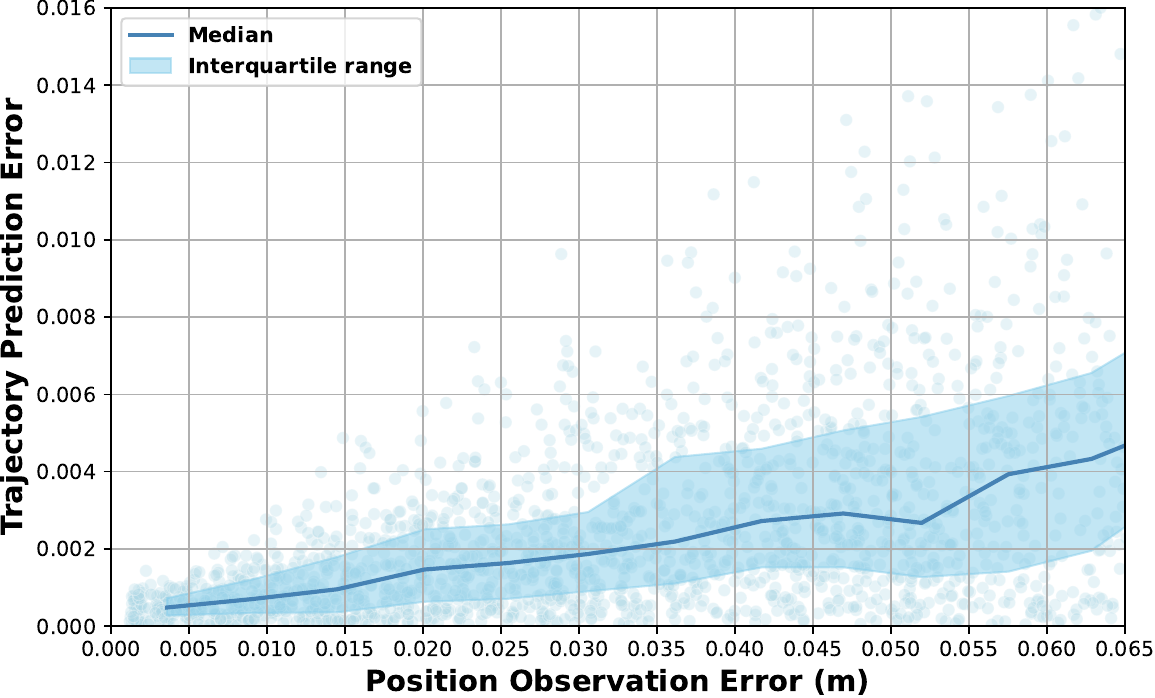}}
    \caption{The distribution of CAVs' trajectory prediction errors for the HV under different position observation errors.}
    \label{fig:offline_trajectory_position}
\end{figure}

Fig. \ref{fig:offline_traj} shows actual trajectories and in cognition of HV and CAV $1$ in one experiment. To make the trajectories distinguishable, the trajectories of CAV $2$ and CAV $3$ are omitted in the figure. It can be seen that the prediction of the trajectory of CAV $1$ in HV's cognition by CAVs is also accurate, while it shows a significant difference from the actual trajectory of CAV $1$, indicating that the proposed method enables CAVs to simulate HV's cognition with high accuracy. Besides, CAVs' position observation errors and position prediction errors for HV at each time step are shown in Fig. \ref{fig:offline_hv_observed_predicted}, indicating that the proposed method can mitigate the influence of observation noise and make the predicted HV's positions more accurate.

\begin{figure}[htbp]
    \centerline{\includegraphics[width=0.75\linewidth, keepaspectratio]{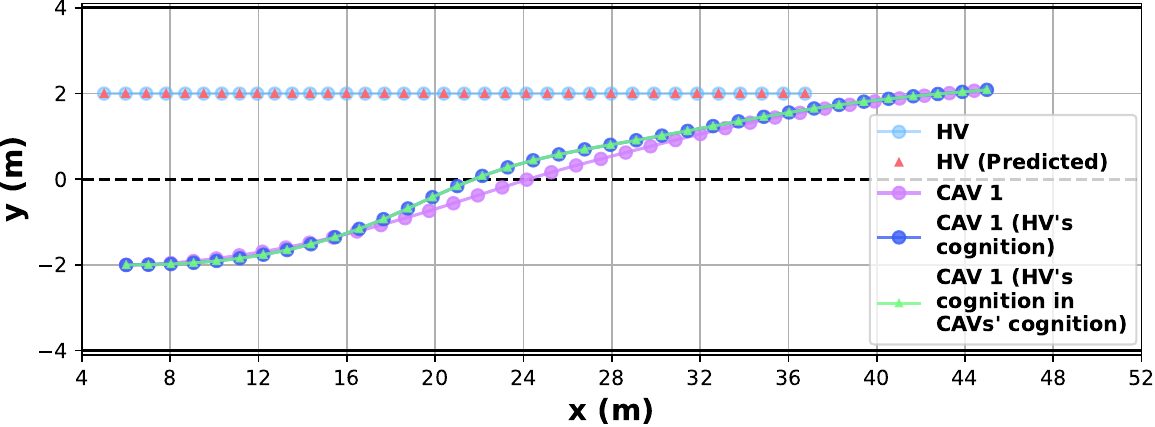}}
    \caption{Vehicles' actual trajectories in fact and in cognition obtained through trajectory generation, prediction, and planning.}
    \label{fig:offline_traj}
\end{figure}

\begin{figure}[htbp]
    \centerline{\includegraphics[width=\linewidth, keepaspectratio]{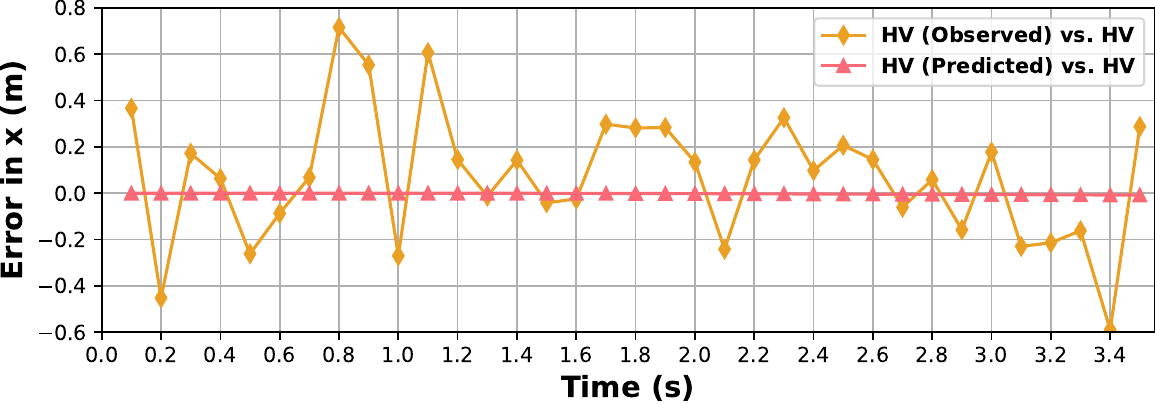}}
    \caption{CAVs' position observation errors and position prediction errors for the HV at each moment.}
    \label{fig:offline_hv_observed_predicted}
\end{figure}

Additionally, we compare the success rate of CAVs' trajectory planning with and without the proposed cognition modeling and intention interpretation algorithm, in order to evaluate the significance of the algorithm in terms of safety. The success rate is the percentage of experiments where both HV and CAVs reach the destination without violating constraints. With the proposed algorithm, we still make CAVs re-predict and re-plan trajectories at the initial moment using the learned parameters, as mentioned before. When CAVs do not use the proposed algorithm, CAVs’ cognition of HV's weight $\theta_{0,C}$ is inaccurate, so in essence, HV and CAVs plan trajectories based on the level 1 hypergame. Specifically, the driving style types of HV and CAVs remain as previously described, with $\theta_{eff,0}=\frac{(1,1,5)^\top}{\|(1,1,5)^\top\|_2}$, and $\theta_{eff,0,C}$ being a random three-dimensional unit vector. The angle between vectors $\theta_{eff,0}$ and $\theta_{eff,0,C}$ follows a uniform distribution $U(\qty{-45}{\degree},\qty{45}{\degree})$. CAV $1$ starts changing lanes at $x=\qty{5}{m}$. Without intention interpretation, the parameter error in CAVs' cognition of HV's weight is defined as $\frac{\|\theta_{eff,0,C}-\theta_{eff,0}\|_2}{\|\theta_{eff,0}\|_2}$, which is the same as the parameter estimation error defined for intention interpretation.

Experimental results show that based on the proposed algorithm, CAVs can safely pass the target location 100 percent. The statistical results of $1,000$ experiments without using the proposed algorithm are shown in Fig. \ref{fig:offline_success}. It can be seen that when HV and CAVs both have misperceptions, CAVs not inferring HV's intention leads to a low success rate. In particular, the success rate remains low even when the parameter error is small. The reason lies in that HV is engaged in a subjective game different from CAVs, where HV's cognition of CAVs' weight $\theta_i^0,i\in\mathcal{N}_{\mathrm{CAV}}$ is biased, and CAVs cannot realize the existence of HV's game when making decisions based on the level 1 hypergame. In this mode, CAVs lack the process of simulating HV's cognition, i.e., Game \ref{de:hv_cognition_in_cavs_cognition} and Problem \ref{de:intention_interpretation_problem}, resulting in a lower success rate. As a consequence, the empirical results demonstrate the superiority of the proposed algorithm regarding the safety of trajectory planning.

\begin{figure}[htbp]
    \centerline{\includegraphics[width=\linewidth, keepaspectratio]{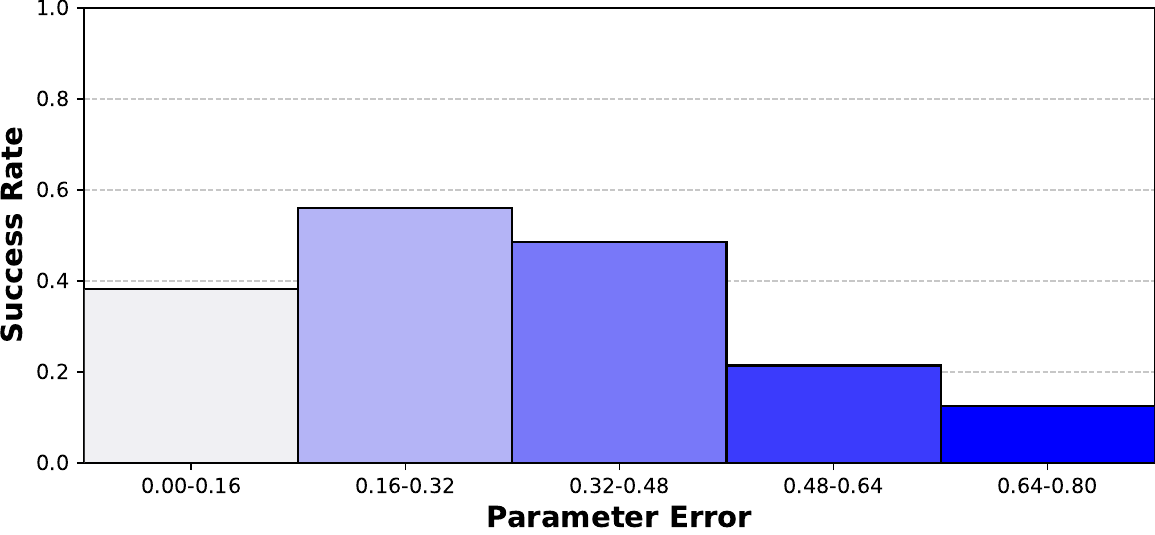}}
    \caption{The success rate of trajectory planning without cognition modeling and intention interpretation under different parameter errors in CAVs' cognition of HV's weight.}
    \label{fig:offline_success}
\end{figure}

\subsection{Online Experiments}
\begin{figure*}[htbp]
    \centerline{\includegraphics[width=0.9\linewidth, keepaspectratio]{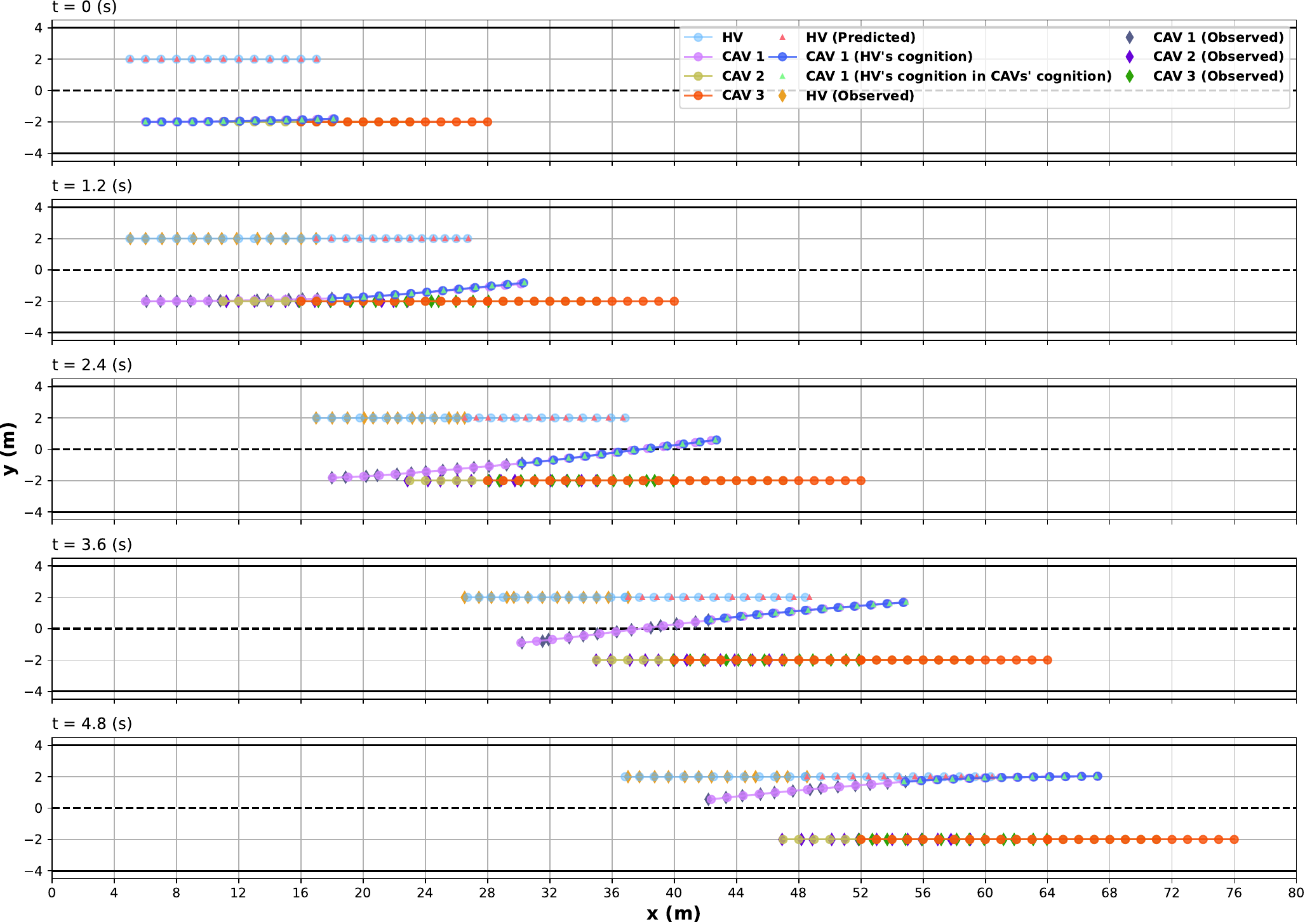}}
    \caption{The trajectories for each vehicle in the online experiment. Those include the actual trajectories of CAVs and HVs across different times, the predicted and observed trajectories of HV as perceived by CAVs, and the trajectories of CAVs as interpreted by HV and their corresponding perceptions by CAVs.}
    \label{fig:online_tra}
\end{figure*}

In the online experiment, we measure the performance of the proposed method in continuously alternating online learning of parameters and decision-making across multiple stages of interaction. The lane-changing process of $\qty{6}{s}$ is divided into five stages of games. The initial $\theta_{0,C}^1$ is set to the typical value of the HV's driving style type. Subsequently, at the beginning of stages $2$ to $5$, CAVs update their estimations of $\theta_{0,C}^t,t=2,3,4,5$, based on the trajectory in the previous stage. The driving style types of HV and CAVs $1$-$3$ are pose-tracking, comfort-oriented, velocity-consistent, and pose-tracking, respectively. The initial speed of each vehicle is $\qty{10}{m/s}$. Both HV and CAVs observe each other's {\it x}-position with Gaussian noise having a mean of $0$ and a standard deviation of $0.05$, while CAVs obtain error-free trajectories through communication. The $x_{\text{ref},i}$ for each stage is obtained by matching the observation points on the complete reference trajectory. In intention interpretation, we set $\kappa=0.3$ and $\omega_{dist}=1$. Due to the limited interaction between CAVs and HVs observed in the first phase, the smoothing term was omitted in the second phase, then the smoothness term is applied at the beginning of stages $3$, $4$, and $5$. The experiment was repeated 50 times.

Fig.~\ref{fig:online_tra} illustrates the experimental scenario and reference trajectories of each vehicle under the online case. In this scenario, the observed trajectories of HV, CAV $1$, CAV $2$, and CAV $3$ are shown at various time steps (e.g., $t = 0$, $1.2$, $2.4$, $3.6$, and $4.8$ seconds). The figure highlights the evolving interactions between the vehicles, where the predicted trajectories align closely with the observed trajectories over time. This demonstrates the effectiveness of the proposed method in real-time applications, providing accurate and reliable trajectory predictions.

Fig.~\ref{fig:online_parameter_trajectory_time} shows the parameter estimation errors and trajectory prediction errors at different times, where parameter estimation errors use the left vertical axis and trajectory prediction errors use the right vertical axis. In the first stage of the game, since CAV $1$ has a small lateral displacement and no collision risk with HV, there is no interaction between the two, and HV travels at a constant speed along the reference trajectory. At this time, for any $\theta_0$, we have $J_0(s_0; \theta_0)=0$. Therefore, CAVs cannot learn the correct weights at $\qty{1.2}{s}$. In the second stage, after the interaction occurs, the parameter estimation error significantly decreases and remains below $\num{0.04}$. Because of the smoothness term, the parameter estimation accuracy at the end of the fourth stage, where interaction is reduced, still maintains the accuracy of the dense interaction in stages $2$ and $3$. The trajectory prediction error also shows a downward trend as the interaction progresses. The experimental results indicate that the proposed method can effectively identify HV's intention during online interaction.

\begin{figure}[htbp]
    \centerline{\includegraphics[width=\linewidth, keepaspectratio]{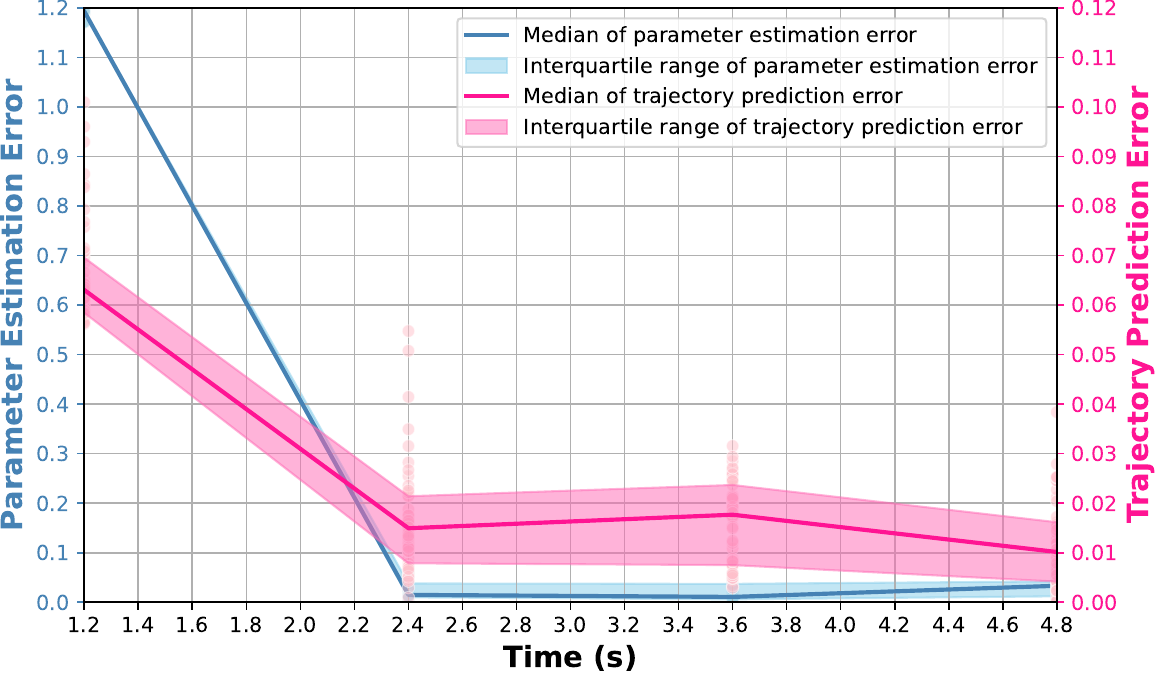}}
    \caption{The distribution of CAVs' parameter estimation errors and trajectory prediction errors for the HV at different moments in online experiments.}
    \label{fig:online_parameter_trajectory_time}
\end{figure}

Comparing the results of Fig.~\ref{fig:online_parameter_trajectory_time} with those of Figs.~\ref{fig:offline_parameter_position} and \ref{fig:offline_trajectory_position}, we can see that the error in online experiments is slightly greater than that in offline experiments. The main reasons
are as follows. Firstly, the prediction horizon of a single game in online experiments is shorter, resulting in a smaller amount of data for learning the weights. Secondly, in online experiments, both HV and CAVs' observations are affected by noise, and the reference trajectory is also obtained by matching noisy observed positions, thus generating additional errors and significantly affecting trajectory prediction.

\begin{figure}[htbp]
    \centerline{\includegraphics[width=\linewidth, keepaspectratio]{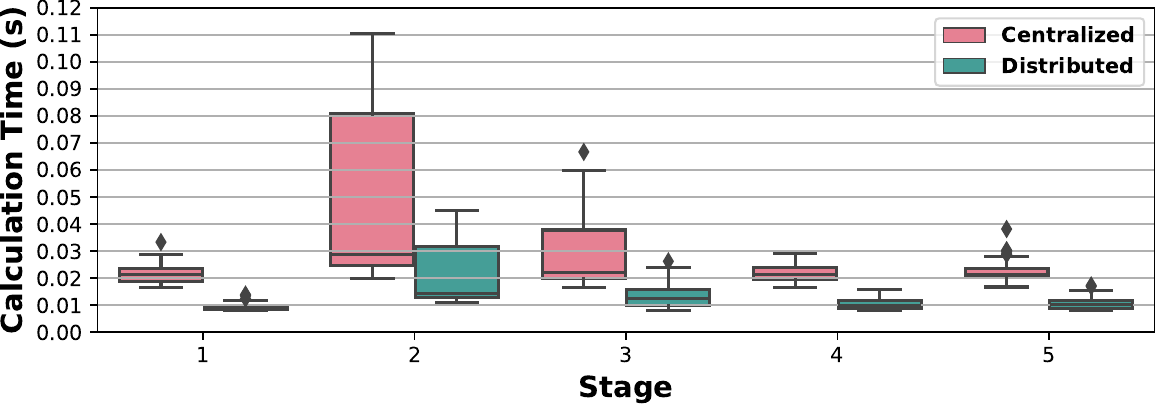}}
    \caption{The computation time of the algorithm's centralized and distributed implementations in each stage in online experiments.}
    \label{fig:online_time}
\end{figure}

Finally, we evaluate the computation time of the algorithm. In particular, we compare the time taken by the proposed distributed algorithm in parameter learning, trajectory prediction, and trajectory planning with its centralized implementation. The distributed implementation of the algorithm is synchronous, with the time determined by the slowest CAV. The centralized implementation of the algorithm refers to the entire computation process of game-solving and intention interpretation being executed by a single CAV or RSU. The program runs on a desktop computer that has Windows 11 installed, an Intel Core i5-10400F CPU, and 16GB of RAM. The time taken by the algorithm in each stage is shown in Fig. \ref{fig:online_time}. It can be seen that the distributed algorithm is significantly more efficient than the centralized algorithm.

\section{Conclusion}
In this paper, we developed a novel framework for the intention interpretation and trajectory planning of HVs within a mixed traffic environment of CAVs. Firstly, we modeled human bounded rationality by incorporating cognitive and perception limitations. Then we proposed a hierarchical cognition modeling method based on hypergame theory to simulate the cognitive relationships between HVs with imprecise cognition and CAVs. To estimate the objective function parameters of HVs, we designed a KKT-based distributed inverse learning algorithm leveraging vehicle-road coordination. Furthermore, we analyzed the cognitive stability of the system and proved that the strategy profile where all vehicles adopt cognitively optimal responses constitutes a hyper Nash equilibrium when CAVs successfully learn the true parameters of HVs (Theorem \ref{thm:hne}). In addition, we extended the intention interpretation and trajectory planning methods to online scenarios, enabling real-time prediction and decision-making. Finally, we conducted simulations in highway lane changing scenarios to demonstrate the accuracy, robustness, and safety of the proposed methods. The results confirmed that our approachcan effectively learn parameters and predicted HV trajectories in both offline and online scenarios, even under noisy observation conditions. Hence, these findings highlighted the potential of our framework to enhance safety and efficiency in mixed traffic systems.




%





\ifCLASSOPTIONcaptionsoff
  \newpage
\fi





\bibliographystyle{IEEEtran}
\bibliography{IEEEabrv,Bibliography}

\vfill


\end{document}